\documentclass[runningheads]{llncs}

\newcommand{\macrospath}{./macros}
\usepackage{ifthen}

\newboolean{talk}
\setboolean{talk}{false}
\newboolean{paper}
\setboolean{paper}{false}

\newboolean{IEEEstyle}
\setboolean{IEEEstyle}{false}
\newboolean{lipicsstyle}
\setboolean{lipicsstyle}{false}
\newboolean{eptcsstyle}
\setboolean{eptcsstyle}{false}
\newboolean{entcsstyle}
\setboolean{entcsstyle}{false}
\newboolean{sigplanstyle}
\setboolean{sigplanstyle}{false}
\newboolean{easychairstyle}
\setboolean{easychairstyle}{false}

\newboolean{lncsstyle}
\setboolean{lncsstyle}{false}

\newboolean{needstheorems}
\setboolean{needstheorems}{false}
\newboolean{withimages}
\setboolean{withimages}{false}
\newboolean{withproofs}
\setboolean{withproofs}{true}
\newboolean{techreport}
\setboolean{techreport}{false}

\newboolean{french}
\setboolean{french}{false}

\setboolean{lncsstyle}{true}
\setboolean{withproofs}{false}
\setboolean{techreport}{true}
\usepackage{stmaryrd}
\usepackage{appendix}
\usepackage{xcolor}
\usepackage{bussproofs}
\usepackage{bbm}
\usepackage[font=small,skip=0pt]{caption}
\usepackage{appendix}
\capturecounter{theorem}
\capturecounter{lemma}
\capturecounter{proposition}
\capturecounter{corollary}
\usepackage{hyperref}

\ifthenelse{\boolean{easychairstyle}}{
\setboolean{needstheorems}{true}}{}

\newcommand{\result}{r}

\makeatletter
\newcommand{\exder}{%
  \def\exderW[##1]{\triangleright_{##1}\ }%
  \def\exderWO{\triangleright\ }%
  \@ifnextchar[\exderW\exderWO%
  }
\makeatother

\newcommand{\app}{{\tt app}\xspace}

\makeatletter
\newcommand{\appresult}{%
  \def\appresultW<##1>{\app_\result^{##1}}%
  \def\appresultWO{\app_\result}%
  \@ifnextchar<\appresultW\appresultWO%
  }
\makeatother

\newcommand{\Id}{{\tt I}}

\newcommand{\tostrat}{\Rew{\stratsym}}

\newcommand{\Terms}{\mathcal T}

\ifthenelse{\boolean{talk}}{}{\usepackage[utf8]{inputenc}}
\ifthenelse{\boolean{french}}{\usepackage[french]{babel}}{
  \ifthenelse{\boolean{lipicsstyle}}{}{\usepackage[english]{babel}}}
\usepackage{amssymb}
\usepackage{amsmath}
\usepackage{graphicx}
\usepackage{cmll}
\usepackage{stmaryrd}
\usepackage{multirow}
\usepackage{fancybox}
\usepackage{xspace}

\ifthenelse{\boolean{IEEEstyle}}{
\setboolean{needstheorems}{true}}{}

\ifthenelse{\boolean{eptcsstyle}}{
\setboolean{needstheorems}{true}}{}

\ifthenelse{\boolean{sigplanstyle}}{
\setboolean{needstheorems}{true}}{}

\ifthenelse{\boolean{needstheorems}}{
\input{\macrospath/ben-macros-thm-environments}}{}

\newcommand{\myproof}[1]{
\ifthenelse{\boolean{withproofs}}{#1}{}
}

\newcommand{\la}[1]{\lambda #1.}

\newcommand{\tm}{t}
\newcommand{\tmtwo}{u}
\newcommand{\tmthree}{r}
\newcommand{\tmfour}{q}

\newcommand{\var}{x}
\newcommand{\vartwo}{y}
\newcommand{\varthree}{z}
\newcommand{\varfour}{w}

\newcommand{\rootRew}[1]{\mapsto_{#1}}

\makeatletter

\newcommand{\Rewbase}{%
  \def\RewbaseW[##1]##2##3{\ {\xrightarrow{##1}}{}_{##2}^{##3}\xspace }%
  \def\RewbaseWO##1##2{\ {\xrightarrow{}}{}_{##1}^{##2}\xspace }%
  \@ifnextchar[\RewbaseW\RewbaseWO%
  }

\newcommand{\Rew}[1]{\rightarrow_{#1}}

\newcommand{\Rewn}[2][*]{%
  \def\RewnW[##1]{\Rewbase[##1]{#2}{#1}}%
  \def\RewnWO{\Rewbase{#2}{#1}}%
  \@ifnextchar[\RewnW\RewnWO%
  }

\makeatother

\renewcommand{\to}{\Rew{}}
\newcommand{\tob}{\mathrel{\Rew{\beta}}}
\newcommand{\rtob}{\rootRew{\beta}}

\newcommand{\lRew}[1]{\; \mbox{}_{#1}{\leftarrow}\ }

\newcommand{\whsym}{wh}

\newcommand{\towh}{\Rew{\whsym}}

\newcommand{\ctxholep}[1]{\langle #1\rangle}
\newcommand{\ctxhole}{\ctxholep{\cdot}}

\newcommand{\ctx}{C}
\newcommand{\ctxtwo}{D}

\newcommand{\ctxp}[1]{\ctx\ctxholep{#1}}

\newcommand{\ctxtwop}[1]{\ctxtwo\ctxholep{#1}}

\newcommand{\apctx}{A}

\newcommand{\nbvctxtwo}[1]{\nbvctxtwo{#1}}

\newcommand{\defeq}{:=}
\newcommand{\grameq}{::=}

\newcommand{\isub}[2]{\{#1{\leftarrow}#2\}}
\newcommand{\esub}[2]{[#1{\leftarrow}#2]}

\newcommand{\llbrace}{\{ \kern -0.27em \vert}
\newcommand{\rrbrace}{\vert \kern -0.27em \}}

\newcommand{\streq}{\equiv}
\newcommand{\fn}[1]{{\tt fn}(#1)}
\renewcommand{\fn}[1]{{\tt names}(#1)}

\renewcommand{\l}{\lambda}

\newcommand{\ie}{{\em i.e.}\xspace}

\newcommand{\ih}{{\emph{i.h.}}\xspace}
\newcommand{\fv}[1]{{\tt fv}(#1)}

\newcommand{\correction}[1]{{\color{blue} {#1}}}
\renewcommand{\correction}[1]{#1}

\newcommand{\red}[1]{{\color{red} {#1}}}

\newcommand{\ignore}[1]{}

\newcommand{\colspace}{@{\hspace{.4cm}}}
\newcommand{\myinput}[1]{\ifthenelse{\boolean{withproofs}}{\input{#1}}{}}

\newcommand{\reflemma}[1]{Lemma~\ref{l:#1}}

\newcommand{\reflemmaeq}[1]{{L.\ref{l:#1}}}
\newcommand{\reflemmaeqp}[2]{{L.\ref{l:#1}.\ref{p:#1-#2}}}

\newcommand{\refthm}[1]{Theorem~\ref{thm:#1}}

\newcommand{\refprop}[1]{Proposition~\ref{prop:#1}}

\newcommand{\refsect}[1]{Sect.~\ref{sect:#1}}

\ifthenelse{\boolean{easychairstyle}}{}{
	\newcommand{\refeq}[1]{(\ref{eq:#1})} 
}
\newcommand{\reffig}[1]{Fig.~\ref{fig:#1}}

\newcommand{\refdef}[1]{Definition~\ref{def:#1}}

\newcommand{\refpoint}[1]{Point~\ref{p:#1}}

\newcommand{\refex}[1]{Example~\ref{ex:#1}}

\newcommand{\set}[1]{\{#1\}}

\newcommand{\nat}{\mathbb{N}}

\newcommand{\size}[1]{|#1|}
\newcommand{\sizep}[2]{|#1|_{#2}}

\newcommand{\env}{\genv}
\newcommand{\envtwo}{\env'}

\newcommand{\genv}{E}

\newcommand{\stack}{S}
\newcommand{\stacktwo}{\stack'}

\newcommand{\estack}{\epsilon}

\newcommand{\exec}{\rho}

\newcommand{\run}{\exec}

\newcommand{\decode}[1]{\underline{#1}}

\newcommand{\stempty}{\epsilon}
\newcommand{\cons}{:}

\newcommand{\deriv}{\ensuremath{e}}

\newcommand{\lo}{{LO}\xspace}

\renewcommand{\lo}{{lo}\xspace}

\newcommand{\tolo}{\Rew{\lo}}

\newcommand{\toll}{\Rew{\llsym}}

\newcommand{\rename}[1]{#1^{\mathtt{R}}}

\renewcommand{\dump}{D}

\newcommand{\mach}{\symfont{M}}

\newcommand{\tomachhole}[1]{\leadsto_{#1}}
\newcommand{\tomach}{\tomachhole{}}

\DeclareMathOperator{\dom}{\symfont{dom}}

\newcommand{\tobo}{\rightarrow_{\mathtt{BO}}}
\renewcommand{\tobo}{\toext}
\newcommand{\extsym}{\symfont{x}}
\newcommand{\toext}{\Rew{\extsym}}
\newcommand{\ltoext}{\lRew{\extsym}}
\newcommand{\aptm}{r}

\newcommand{\lam}[2]{\lambda#1.\,#2}
\newcommand{\soctx}{S}
\newcommand{\boctx}{B}
\renewcommand{\soctx}{R}

\renewcommand{\boctx}{E}

\newcommand{\boctxp}[1]{\boctx\ctxholep{#1}}

\newcommand{\name}{\alpha}
\newcommand{\nametwo}{\beta}
\newcommand{\namethree}{\gamma}

\newcommand{\pool}{P}
\newcommand{\pooltwo}{\pool'}

\newcommand{\place}{j}

\newcommand{\cplace}[3]{(#2,#3)_{#1}}

\newcommand{\eempty}{\epsilon}
\newcommand{\pempty}{\epsilon}
\renewcommand{\pempty}{\emptyset}

\newcommand{\soapr}{\mathbb{S}}
\newcommand{\boapr}{\mathbb{B}}
\renewcommand{\soapr}{\mathbb{R}}
\newcommand{\soaprtwo}{\mathbb{R'}}
\renewcommand{\boapr}{\mathbb{A}}
\newcommand{\boaprtwo}{\mathbb{A}'}
\newcommand{\boctxhole}[1]{\ctxhole_{#1}}
\newcommand{\boctxholep}[2]{\langle#2\rangle_{#1}}

\newcommand{\SEP}{\hspace{2pt}|\hspace{2pt}}
\newcommand{\state}{s}
\newcommand{\statetwo}{\state'}
\newcommand{\statethree}{\state''}
\newcommand{\statefour}{\state'''}

\newcommand{\examstate}[3]{\llbracket{#1}\SEP{#2}\SEP{#3}\rrbracket}
\newcommand{\mamstate}[3]{({#1}\SEP{#2}\SEP{#3})}

\newcommand{\ntm}{\mathbbm{C}}
\newcommand{\ntmtwo}{\mathbbm{C}'}

\newcommand{\sub}[2]{\{{#1}\defeq{#2}\}}
\newcommand{\tsub}[2]{\langle#1\defeq#2\rangle}
\renewcommand{\tsub}[2]{\langle #2\rangle_{#1}}
\newcommand{\decpool}[2]{\decode{#1}_{#2}}
\newcommand{\decenv}[2]{\decode{#1}_{#2}}

\newcommand{\decst}[2]{\decode{#1}_{#2}}

\newcommand{\mydots}{..}
\newcommand{\EXAM}{EXAM\xspace}
\newcommand{\LEXAM}{Leftmost EXAM\xspace}

\newcommand{\SEXAM}{Set EXAM\xspace}
\newcommand{\seasym}{\mathtt{sea}}
\newcommand{\toexam}{\tomachhole{\mathtt{\EXAM}}}
\newcommand{\tolexam}{\tomachhole{\mathtt{L\EXAM}}}

\newcommand{\toexamap}{\tomachhole{\seasym_@}}
\newcommand{\toexamlam}{\tomachhole{\seasym_\lambda}}
\newcommand{\toexambeta}{\tomachhole{\beta}}
\newcommand{\subsym}{\mathtt{sub}}
\newcommand{\toexamsub}{\tomachhole{\subsym}}
\newcommand{\seavarsym}{\seasym_{\mathcal{V}}}
\newcommand{\toexamvar}{\tomachhole{\seavarsym}}

\newcommand{\toexamvartwo}{\tomachhole{\seasym_{\mathcal{V}2}}}

\newcommand{\symfont}[1]{\mathtt{#1}}
\newcommand{\osym}{\symfont{o}}
\newcommand{\tomacho}{\tomachhole{\osym}}

\newcommand{\tomachine}{\tomachhole\mach}

\newcommand{\tomachb}{\tomachhole{\beta}}

\newcommand{\nfo}[1]{\symfont{nf}_\osym(#1)} 

\newcommand{\compilrel}[2]{#1\triangleleft#2}

\newcommand{\sizebeta}[1]{\sizep{#1}{\beta}}

\newcommand{\execp}{\exec'}
\newcommand{\execpp}{\exec''}

\newcommand{\derivp}{\deriv'} 
\newcommand{\derivpp}{\deriv''} 

\newcommand{\tmtwop}{\tmtwo'}

\newcommand{\omeas}[1]{\sizep{#1}\osym}

\newcommand{\neu}{n}

\newcommand{\nf}{f}

\newcommand{\addr}{\mathbbm{a}}
\newcommand{\addrtwo}{\addr'}

\newcommand{\staddr}[2]{#1|_{#2}}
\newcommand{\ems}{\epsilon}
\newcommand{\pop}{\overset{\symfont{sel}}{\leftharpoondown}}
\newcommand{\push}{\overset{\symfont{dro}}{\rightharpoondown}}
\newcommand{\add}{\overset{\symfont{add}}{\rightharpoondown}}
\renewcommand{\add}{\overset{\symfont{add}}{\rightharpoondown}}
\newcommand{\pro}{\add}
\newcommand{\lctx}{L}
\newcommand{\lctxtwo}{\lctx'}

\newcommand{\lctxp}[1]{\lctx\ctxholep{#1}}

\newcommand{\llsym}{\ell\ell}

\newcommand{\nctx}{N}
\newcommand{\nctxtwo}{\nctx'}

\newcommand{\States}{\symfont{States}}
\renewcommand{\tostrat}{\Rew{\symfont{str}}}

\newcommand{\job}{\place}
\newcommand{\jobn}[1]{\job_{#1}}

\newcommand{\supp}[1]{\symfont{supp}(#1)}
\newcommand{\drop}{\push}
\newcommand{\new}[1]{\symfont{new}(#1)}
\newcommand{\names}[1]{\symfont{names}(#1)}
\newcommand{\suppset}{X}
\newcommand{\suppsettwo}{\suppset'}

\newcommand{\skeval}{\textcolor{blue}{\blacktriangledown}}
\newcommand{\skback}{\textcolor{red}{\blacktriangle}}

\newcommand{\absstack}{A}

\newcommand{\ifreport}[2]{\ifthenelse{\boolean{techreport}}{#1}{#2}}
\usepackage{tikz}
\usetikzlibrary{calc}
\usetikzlibrary{matrix,arrows}
\usetikzlibrary{decorations.shapes}
\usetikzlibrary{decorations.text}
\usetikzlibrary{decorations.pathmorphing}
\usetikzlibrary{decorations.markings}
\usetikzlibrary{positioning}

\tikzset{
node distance=1.3cm, auto,
every node/.style={font=\scriptsize },
ocenter/.style={baseline={([yshift=-.5ex, xshift=-.5ex]current bounding box)}},  
labelBeginAbove/.style={postaction={decorate,decoration={markings,mark=at position 0 with {\node[inner sep= 0.6pt, above=1pt]{\tiny #1};}} } },
labelBeginBelow/.style={postaction={decorate,decoration={markings,mark=at position 0 with {\node[inner sep= 0.6pt, below=1pt]{\tiny #1};}}}},
labelEndAbove/.style={postaction={decorate,decoration={markings,mark=at position 1 with {\node[inner sep= 0.6pt, above=1pt]{\tiny #1};}}}},
labelEndBelow/.style={postaction={decorate,decoration={markings,mark=at position 1 with {\node[inner sep= 0.6pt, below=1pt]{\tiny #1};}}}},
labelEndRight/.style={postaction={decorate,decoration={markings,mark=at position 1 with {\node[inner sep= 0.6pt, right=1pt]{\tiny #1};}}}},
labelEndLeft/.style={postaction={decorate,decoration={markings,mark=at position 1 with {\node[inner sep= 0.6pt, left=1pt]{\tiny #1};}}}}
}

\newcommand{\verticaldiagramdistance}{25pt}

\begin{document}
\title{A Diamond Machine for Strong Evaluation}
\author{Beniamino Accattoli\inst{1}\orcidID{0000-0003-4944-9944} \and Pablo Barenbaum\inst{2,3}}
\authorrunning{Accattoli and Barenbaum}
\institute{Inria \& LIX, \'Ecole Polytechnique, UMR 7161, Palaiseau, France
\and
Universidad Nacional de Quilmes (CONICET), Bernal, Argentina
\and
CONICET-Universidad de Buenos Aires, Instituto de Ciencias de la Computación, Argentina
}

\maketitle

\begin{abstract}
Abstract machines for strong evaluation of the $\lambda$-calculus enter into arguments and have a set of transitions for backtracking out of an evaluated argument. We study a new abstract machine which avoids backtracking by splitting the run of the machine in smaller \emph{jobs}, one for argument, and that jumps directly to the next job once one is finished.

Usually, machines are also deterministic and implement deterministic strategies. Here we weaken this aspect and consider a light form of non-determinism, namely the \emph{diamond property}, for both the machine and the strategy. For the machine, this introduces a modular management of jobs, parametric in a scheduling policy. We then show how to obtain various strategies, among which leftmost-outermost evaluation, by instantiating in different ways the scheduling policy.

\end{abstract}

\section{Introduction}
An abstract machine for the $\l$-calculus, or for one of its extensions, is an implementation schema for a fixed evaluation strategy 
$\tostrat$ with sufficiently atomic operations (accounting for the \emph{machine} part) and without too many implementative details 
(accounting for the \emph{abstract} part). An abstract machine for $\tostrat$, ideally, refines the reduction of $\tostrat$-redexes realizing the following three tasks:
\begin{enumerate}
  \item \emph{Search}: searching for $\tostrat$-redexes;
  \item \emph{Names}: avoiding variable captures through some mechanism implementing $\alpha$-conversion, or allowing one to avoid $\alpha$-conversion altogether;
   \item \emph{Substitution}: refining meta-level substitution with an approximation based on delaying the substitution, which boils down to adopting a form of sharing, and replacing one variable occurrence at a time, in a demand-driven fashion.
\end{enumerate}
These tasks are usually left to \emph{meta-level operations} in $\l$-calculi, meaning that they happen outside the syntax of the 
calculus itself, in a black-box manner. The role of abstract machines is to explicitly take care of these aspects, or at least of some of them, reifying them from the meta-level to the object-level. Concretely, this is obtained by enriching the specification of the operational semantics with dedicated data structures. Additionally, such a specification is usually designed as to be efficient, and usually the evaluation strategy $\tostrat$ is deterministic. 

\paragraph{Search, Backtracking, and Jumping.} A first motivation of this paper is obtaining a better understanding of the search mechanism of abstract machines. When pushed at the meta-level, search is usually specified via deduction rules or via a grammar of evaluation contexts, assuming that, at each application of a rewriting rule, the term is correctly split into an 
evaluation context and a redex. The meta-level aspect is the fact that the process of splitting the term (or applying the deductive rules)
is not taken into account as an operation of the calculus.

For simple evaluation strategies such as, for instance, weak call-by-name in the pure $\l$-calculus (also known as \emph{weak head reduction}), abstract machines (such as the Krivine or the Milner abstract machine) have one search transition for every production of the evaluation contexts for the meta-level definition of search. For less simple strategies, the searching for redexes by the machine often also accounts for the further mechanism of \emph{backtracking search}, that is not visible in the operational semantics, not even at meta-level. Such a form of search---which is completely unrelated to \emph{backtracking-as-in-classical-logic}---happens when the machine has finished evaluating a sub-term and needs to backtrack to retrieve the next sub-term to inspect. Typically, this happens when implementing strategies that evaluate arguments, and in particular for strategies evaluating to \emph{strong} normal form (that is, also under abstraction), the paradigmatic and simplest example of which is leftmost(-outermost\footnote{To ease the language, in the paper we shorten \emph{leftmost-outermost} to \emph{leftmost}.}) evaluation. 

As an example, let $\nf$ be a strong normal form and consider executing a machine for leftmost evaluation on $\la\var\var\nf\tm$. The machine would go under $\la\var$, find the head variable $\var$ and then start searching for a $\beta$-redex in $\nf$. Since there are none, the machine arrives at the end of $\nf$ and then it usually \emph{backtracks} through the structure of $\nf$, as to exit it and then start searching inside $\tm$. Backtracking search is natural when one sees the searching process as a \emph{walk} over the code, moving only between \emph{adjacent} constructors. This is the dominating approach in the design of abstract machines for $\l$-calculi, since the entirety of the (small) literature on machines for strong evaluation adopts it \cite{DBLP:journals/lisp/Cregut07,DBLP:conf/ppdp/Garcia-PerezNM13,DBLP:conf/aplas/AccattoliBM15,DBLP:journals/jfp/Garcia-PerezN19,DBLP:conf/rta/BiernackaC19,DBLP:conf/aplas/BiernackaBCD20,DBLP:conf/lics/AccattoliCC21,DBLP:conf/ppdp/BiernackaCD21,DBLP:journals/pacmpl/BiernackaCD22}. 

There is, however, a natural alternative approach which is \emph{saving} the position where one would next backtrack, and then directly \emph{jumping} to that position, instead of walking back to it.
In this paper, we explore how to avoid backtracking search by adopting a form of \emph{jumping search}. The idea is embarrassingly simple: creating a new \emph{job} when an argument ready to be evaluated is found, adding it to a pool of jobs; then when the active job terminates, jumping to another job in the pool, without backtracking out of the terminated one.

\paragraph{Diamond Non-Determinism.} A second motivation of the paper is to understand how to deal with diamond non-determinism at the level of machines. It is often the case that a deterministic strong strategy can be relaxed as to be non-deterministic. For instance, on the head normal form $\var\,\tm\,\tmtwo\,\tmthree$ leftmost evaluation would first evaluate $\tm$, then $\tmtwo$, and then $\tmthree$. But the evaluations of $\tm$, $\tmtwo$, and $\tmthree$ are in fact independent, so that one could evaluate them in any order. One could even interleave their evaluations, as it is done for instance by the least level strategy, a non-deterministic strong strategy coming from the linear logic literature, introduced---we believe---by Girard \cite{DBLP:journals/iandc/Girard98} and studied for instance by de Carvalho et al. \cite{DBLP:journals/tcs/CarvalhoPF11} on proof nets and by Accattoli et al. \cite{DBLP:conf/aplas/AccattoliFG19} on the $\l$-calculus. 

Such a form of non-determinism is benign, as it does not affect the result, nor the length of the evaluation. 
Abstractly, it is captured by the \emph{diamond property} \correction{(here defined following Dal Lago and Martini \cite{DBLP:journals/tcs/LagoM08}, while Terese \cite{terese} defines it more restrictively, without requiring $u_1 \neq u_2$)}
, the strongest form of confluence:
\begin{center}
\begin{tabular}{c@{\hspace{.5 cm}}c@{\hspace{.5 cm}}ccc}
\begin{tikzpicture}[ocenter]
  \node (s) {\small $t$};
  \node at (s) [right =\verticaldiagramdistance] (s1){\small$u_1$};
  \node at (s) [below =\verticaldiagramdistance] (s2){\small$u_2$};
  \draw[->,  cyan] (s) to (s1);
  \draw[->, cyan] (s) to (s2);
\end{tikzpicture}
&
and $u_1\neq u_2$ imply $\exists \red{\tmthree}$ s.t.
&
\begin{tikzpicture}[ocenter]
  \node (s) {\small $t$};
  \node at (s) [right =\verticaldiagramdistance] (s1){\small$u_1$};
  \node at (s) [below =\verticaldiagramdistance] (s2){\small$u_2$};
  \node at (s1|-s2) (t) {\small$\red{\tmthree}$};
  \draw[->,  cyan] (s) to (s1);
  \draw[->, cyan] (s) to (s2);
  \draw[->, dashed,  red] (s2) to  (t);
  \draw[->, dashed,  red] (s1) to  (t);
\end{tikzpicture}
\end{tabular}
\end{center}
What makes it stronger than confluence is that both the opening span from $\tm$ and the closing one to $\tmthree$ are made of \emph{single} steps, not of sequences of steps.

The diamond property can be seen as a liberal form of determinism, because---when it holds---all reductions to normal form have the same length, and if one such reduction exists then there are no diverging reductions.

\paragraph{External Strategy and Machine.} Here we introduce a relaxed, diamond version of the leftmost strategy, which we deem \emph{external strategy}. We are inspired by Accattoli et al. \cite{DBLP:conf/lics/AccattoliCC21}, who study a similar strategy for strong call-by-value. 

\correction{Diamond strategies can be seen as uniform frameworks capturing different deterministic strategies (for instance the leftmost and the least level strategy). 
Accordingly, a non-deterministic machine implementing a diamond strategy would factor out the commonalities of different deterministic machines.} 

We then design a machine for the external strategy, the \emph{EXternal Abstract Machine} (\EXAM) by building over the jumping search explained above. The idea, again, is very simple. It amounts to relaxing the scheduling of jobs from the pool, by allowing the machine to non-deterministically select whatever unfinished job at each step, instead of having to terminate the active job and then having to move to the next one in the pool.

In fact, we go one step further. We define a \emph{pool interface} and the definition of the \EXAM is \emph{abstract} in that it only refers to the interface. Then one can define different \emph{pool templates} that implement various scheduling policies for jobs. The external strategy is implemented when adopting the most general, non-deterministic template. By only replacing the template, we show how the same machine can also implement the leftmost strategy. At the end of the paper, we also quickly overview a template for the least level strategy, as well as one for a fair strategy.


\paragraph{Related Work.}  This work adopts terminologies and techniques from the work on abstract machines by Accattoli and coauthors \cite{DBLP:conf/icfp/AccattoliBM14,DBLP:conf/aplas/AccattoliBM15,DBLP:conf/wollic/Accattoli16,DBLP:conf/ppdp/AccattoliB17,DBLP:conf/ppdp/AccattoliCGC19,DBLP:journals/scp/AccattoliG19,DBLP:conf/lics/AccattoliCC21}, and in particular refines their machine for leftmost evaluation \cite{DBLP:conf/aplas/AccattoliBM15}. They focus on the complexity of machines, while here we focus on \emph{search} and ignore complexity, since their work shows that search has \emph{linear} cost (in the time cost model, i.e. the number of $\beta$-steps) and thus it does not affect the asymptotic behavior, which is instead linked to how the machine realizes the orthogonal \emph{substitution} task. The study of machines for strong evaluation is a blind spot of the 
field, despite the relevance for the implementation of proof assistants. The few studies in the literature have all been cited above. Search
for abstract machines is related to Danvy and Nielsen's (generalized) \emph{refocusing} \cite{Danvy04refocusingin,DBLP:conf/rta/BiernackaCZ17}, which however has never dealt with the jumping search introduced here. General non-deterministic machines are studied by Biernacka et al. \cite{DBLP:conf/concur/BiernackaBLS22}, but the setting is different as they are not diamond.

\paragraph{Proofs.} Proofs are in \ifthenelse{\boolean{techreport}}{the Appendix.}{the technical report \cite{accattoli2023diamond}.}
\section{Normal Forms and the Importance of Being External}
\paragraph{Basics of $\l$.}The set $\Terms$ of untyped $\lambda$-terms is defined by
$\tm \grameq \var \mid \lam{\var}{\tm} \mid \tm\,\tm$.
The capture-avoiding substitution of $\var$ by $\tmtwo$ in $\tm$
is written $\tm\sub{\var}{\tmtwo}$.
The relation of $\beta$-reduction at the root ${\rtob} \subseteq \Terms \times \Terms$
is defined by $(\lam{\var}{\tm})\,\tmtwo \rtob \tm\sub{\var}{\tmtwo}$.

We shall use various notions of contexts, which are terms with a single occurrence of a free variable
$\ctxhole$, called a {\em hole}.
If $\ctx$ is a context, $\ctxp{\tm}$ denotes the \emph{plugging} of $\tm$ in $\ctx$ which is 
the textual substitution of $\ctxhole$ by $\tm$ in $\ctx$. Plugging might capture variables, for instance if $\ctx = \la\var\la\vartwo\ctxhole$ then $\ctxp{\var\vartwo} = \la\var\la\vartwo\var\vartwo$. Note that instead one would have $(\la\var\la\vartwo\varthree)\sub\varthree{\var\vartwo} = \la{\var'}\la{\vartwo'}\var\vartwo$.

The relation of $\beta$-reduction ${\tob} \subseteq \Terms \times \Terms$
is the context closure of $\rtob$, i.e. $\ctxp\tm \tob \ctxp\tmtwo$ if $\tm \rtob\tmtwo$, compactly noted also as ${\tob} \defeq \ctxp{\rtob}$. \correction{An evaluation $\deriv: \tm \tob^* \tmtwo$ is a possibly empty sequence of $\beta$-steps.}

\begin{proposition}[Normal forms]
$\beta$-Normal forms are described by:
\begin{center}
$\begin{array}{r\colspace rcl\colspace\colspace\colspace r\colspace rcl}
\textsc{Neutral terms} & \neu & \grameq & \var \mid \neu \nf
&
\textsc{Normal forms} & \nf & \grameq &\neu \mid \la\var\nf
\end{array}$
\end{center}
\end{proposition}

\paragraph{Weak Head Reduction and External Redexes.} The simplest evaluation strategy is weak head reduction $\towh$, which is obtained as the closure $\apctx\ctxholep\rtob$ of the root $\beta$-rule $\rtob$ by the following notion of applicative contexts:
\begin{center}$
  \begin{array}{l@{\hspace{10pt}}llcccccc}
    \textsc{Applicative contexts}
    & \apctx & ::= & \ctxhole      & \mid & \apctx\,\tm.
  \end{array}$
\end{center}
Example: $(\la\var\tm)\tmtwo\tmthree \towh \tm\isub\var\tmtwo\tmthree$.
Weak head reduction is deterministic. It fails to compute $\beta$-normal forms because it does not reduce arguments nor abstraction bodies, indeed $\tmthree((\la\var\tm)\tmtwo)\not\towh \tmthree(\tm\isub\var\tmtwo)$ and $\la\vartwo((\la\var\tm)\tmtwo)\not\towh \la\vartwo\tm\isub\var\tmtwo$. 

The key property of the weak head redex is that it is \emph{external}, a key concept from the advanced rewriting theory of the $\l$-calculus studied by many authors \cite{DBLP:conf/birthday/HuetL91,DBLP:conf/birthday/HuetL91a,DBLP:conf/popl/Maranget91,DBLP:conf/lics/GonthierLM92,phdmellies,DBLP:journals/iandc/BarendregtKKS87,Boudol86,DBLP:conf/rta/Oostrom99,terese,DBLP:conf/popl/AccattoliBKL14}. Following Accattoli et al. \cite{DBLP:conf/popl/AccattoliBKL14}, the intuition is that a redex $R$ of a term $\tm$ is \emph{external} if:
\begin{enumerate}
\item \emph{Action constraint}: no other redex in $\tm$ can act on (that is, can erase or duplicate) $R$, and
\item \emph{Hereditary clause}: the same it is true, hereditarily, for all the residuals of $R$ after any other redex.
\end{enumerate}
In $\delta (\Id \var)$, where $\delta\defeq \la\vartwo\vartwo\vartwo$ is the duplicator combinator and $\Id \defeq \la\varthree\varthree$ is the identity combinator, the redex $\Id\var$ can be duplicated by $\delta$, so it is not external because the action constraint is not respected. In $\Id \delta (\Id \var)$, instead, the redex $\Id\var$ respects the action constraint, because $\Id\var$ is an outermost redex, and yet $\Id\var$ is not external because it does not validate the hereditary clause: its only residual after the step $\Id \delta (\Id \var) \tob \delta (\Id \var)$ can be duplicated by $\delta$.

Defining external redexes requires the introduction of the theory of residuals, which is heavy and beyond the scope of this paper. The intuition behind it however guides the study in this paper, and we consider it a plus---rather than a weakness---that this can be done  circumventing the theory of residuals.

\paragraph{Leftmost Reduction.} One way to extend weak head reduction to compute $\beta$-normal forms as to always reduce external redexes is provided by \emph{leftmost-outermost reduction} $\tolo$ (shortened to \emph{leftmost}). The definition relies on the notion of neutral term $\neu$ used to describe normal forms, and it is given by the closure $\lctxp\rtob$ of root $\beta$ by the following notion of leftmost contexts, defined by mutual induction with neutral contexts:
\begin{center}$
  \begin{array}{l@{\hspace{10pt}}llccccc@{\hspace{23pt}}cccccc}
    \textsc{Neutral ctxs}
    & \nctx & ::= & \ctxhole & \mid & \neu \lctx        & \mid & \nctx \tm
    &
        \textsc{Leftmost ctxs}
    & \lctx & ::= & \nctx  & \mid & \la\var\lctx 

  \end{array}$
\end{center}
Some examples: $\vartwo((\la\var\tm)\tmtwo)\tolo \vartwo(\tm\isub\var\tmtwo)$ and $\la\vartwo((\la\var\tm)\tmtwo)\tolo \la\vartwo\tm\isub\var\tmtwo$ but $\Id\vartwo((\la\var\tm)\tmtwo)\not\tolo \Id\vartwo(\tm\isub\var\tmtwo)$.
Leftmost reduction is deterministic and \emph{normalizing}, that is, if $\tm$ has a reduction to normal form $\nf$ then leftmost reduction reaches $\nf$. Formally, if $\tm \tob^*\nf$ with $\nf$ normal then $\tm\tolo^*\nf$---for a recent simple proof of this classic result see Accattoli et al. \cite{DBLP:conf/aplas/AccattoliFG19}. The normalization property can be seen as a consequence of the fact that the strategy reduces only external redexes. Note that the outermost strategy (that reduces redexes not contained in any other redex) is instead not normalizing, as the following $\Omega$ redex is outermost (but not leftmost), where $\Omega \defeq \delta\delta$ is the paradigmatic looping $\l$-term:
\begin{equation}
(\lam{\var}{\lam{\vartwo}{\var}})\,\varthree\,\Omega
  \tob
  (\lam{\var}{\lam{\vartwo}{\var}})\,\varthree\,\Omega
  \tob
  \hdots
  \label{eq:outermost-not-normalizing}
  \end{equation}
  The key point is that the outermost strategy does reduce redexes that cannot be acted upon, but it does not satisfy the hereditary clause in the intuitive definition of external redex given above, for which one also needs an additional requirement such as selecting the leftmost redex among the outermost ones.

\paragraph{External Reduction.} It is possible to define a strategy relaxing leftmost reduction, still reducing only external redexes, what we call \emph{external reduction}. The definition uses the auxiliary notions of rigid terms and contexts, plus the applicative contexts $\apctx$ used for weak head reduction. The terminology is inspired by Accattoli et al.'s similar strategy for strong call-by-value evaluation \cite{DBLP:conf/lics/AccattoliCC21}.

\begin{definition}
The following categories of terms and contexts are defined mutually inductively by the grammar:
\[
  \begin{array}{r@{\hspace{10pt}}ccccccc@{\hspace{23pt}}c@{\hspace{10pt}}cccccccccccc}
    \textsc{Rigid terms}
    & \aptm  & ::= & \var          & \mid & \aptm\,\tm
  \\
    \textsc{Rigid ctxs}
    & \soctx & ::= & \ctxhole & \mid & \aptm\,\boctx & \mid & \soctx\,\tm
  &
      \textsc{External ctxs} & \boctx & ::=  & \soctx      & \mid & \lam{\var}{\boctx}
  \end{array}
\]
{\em External reduction} is the rewriting relation ${\tobo} \subseteq \Terms \times \Terms$
on $\lambda$-terms defined as the closure of root $\beta$-reduction under external contexts,
that is,
${\tobo} \defeq \boctx\ctxholep{\rtob}$.
\end{definition}

Alternative streamlined definitions for these notions are:
\begin{center}$
  \begin{array}{lllll@{\hspace{23pt}}lllllllllll}

    \aptm
      & ::=  & \var\,\tm_1\hdots\tm_n

  \\
\soctx
      & ::=  & \ctxhole\,\tm_1\hdots\tm_n & \mid & \var\,\tmtwo_1\hdots\tmtwo_m\,\boctx\,\tm_1\hdots\tm_n
  &
\boctx & ::=  &  \lam{\var_1\hdots\var_k}{\soctx}
  \end{array}$
\end{center}

As proved below, the leftmost strategy is a special case of the external one. The converse does \emph{not} hold: $\tm = \var (\Id \vartwo) (\Id\varthree) \toext \var (\Id\vartwo)\varthree=\tmtwo $ but $\tm \not\tolo \tmtwo$. Instead, $\tm \tolo \var  \vartwo (\Id\varthree)=\tmthree$. Note also a case of diamond: $\tm \toext \tmthree$ and $\tmthree \toext \var\vartwo\varthree \ltoext \tmtwo$.
\begin{toappendix}
\begin{proposition}[Properties of external reduction]
\label{prop:ext-diamond}
\begin{enumerate}
\item \emph{Leftmost is external}: if $\tm\tolo\tmtwo$ then $\tm\toext\tmtwo$.
\item \emph{External diamond}:
if $\tmtwo\ltoext \cdot \toext\tmthree$ with $\tmtwo \neq \tmthree$ then $\tmtwo \toext \cdot \ltoext \tmthree$.
\end{enumerate}
\end{proposition}
\end{toappendix}


\section{Preliminaries: Abstract Machines}
\label{sect:prel-abs-mach}
\paragraph{Abstract Machines Glossary.}  Abstract machines manipulate \emph{pre-terms}, that is, terms without implicit $\alpha$-renaming. In this paper, an \emph{abstract 
machine} is a quadruple $\mach = (\States, \tomach, \compilrel\cdot\cdot, \decode\cdot)$ the component of which are as follows.
\begin{itemize}

\item \emph{States.} A state $\state\in\States$ is composed by the \emph{active term} $\tm$, and some data structures. Terms in states are actually pre-terms.

\item  \emph{Transitions.} The pair $(\States, \tomach)$ is a transition 
system with transitions $\tomach$ partitioned into \emph{$\beta$-transitions} $\tomachb$ (usually just one), that are meant to correspond to $\beta$-steps, and \emph{overhead transitions} $\tomacho$, that take care of the various tasks of the machine (searching, substituting, and $\alpha$-renaming).

\item \emph{Initialization.} \correction{The component $\compilrel{}{}\subseteq\Lambda\times\States$ is the \emph{initialization relation} associating  $\l$-terms to 
initial states. It is a \emph{relation} and not a function because $\compilrel\tm\state$ maps a $\l$-term $\tm$ (considered modulo $\alpha$) to a state $\state$ having a \emph{pre-term representant} of $\tm$ (which is not modulo $\alpha$) as active term. Intuitively, any two states $\state$ and $\statetwo$ such that $\compilrel\tm\state$ and $\compilrel\tm\statetwo$ are $\alpha$-equivalent.} A state $\state$ is \emph{reachable} if it can be reached starting from an initial state, that is, \correction{if $\statetwo \tomach^*\state$ where $\compilrel\tm\statetwo$ for some $\tm$ and $\statetwo$, which we abbreviate using $\compilrel\tm\statetwo \tomach^*\state$}.

\item \emph{Read-back.} The read-back function $\decode\cdot:\States\to\Lambda$ turns reachable states into 
$\l$-terms and satisfies the \emph{initialization constraint}: \correction{if $\compilrel\tm\state$ then $\decode{\state}=_\alpha\tm$}.
\end{itemize}

\paragraph{Further Terminology and Notations.} A state is \emph{final} if no transitions apply.
 A \emph{run} $\run: \state \tomach^*\statetwo$ is a possibly empty finite sequence of transitions, the length of which is noted 
$\size\run$\correction{; note that the first and the last states of a run are not necessarily initial and final}. 
If $a$ and $b$ are transitions labels (that is, $\tomachhole{a}\subseteq \tomach$ and 
$\tomachhole{b}\subseteq \tomach$) then $\tomachhole{a,b} \defeq \tomachhole{a}\cup \tomachhole{b}$ and $\sizep\run a$ 
is the number of $a$ transitions in $\run$.

\paragraph{Well-Namedness and Renaming.} \correction{For the machines at work in this paper, the pre-terms in initial states shall be \emph{well-named}, that is, they have pairwise distinct bound names; for instance $(\la\var\var)(\la\vartwo\vartwo\vartwo)$ is well-named while $(\la\var\var)(\la\var\var\var)$ is not. }
We shall also write $\rename{\tm}$ in a state $\state$ for a \emph{fresh well-named renaming} of $\tm$,
\ie $\rename{\tm}$ is $\alpha$-equivalent to $\tm$, well-named, and its bound variables
are fresh with respect to those in $\tm$ and in the other components of $\state$. 

\paragraph{Implementation Theorem, Abstractly.} We now formally define the notion of a machine implementing a strategy.
\begin{definition}[Machine implementation]
A machine $\mach=(\States, \tomach, \compilrel\cdot\cdot, \decode\cdot)$ \emph{implements a strategy} $\tostrat$ when given a $\l$-term $\tm$ the following 
holds:\label{def:implem}
\begin{enumerate}
\item \emph{Runs to evaluations}: for any $\mach$-run $\exec: \compilrel\tm\state \tomachine^* \statetwo$ there exists a 
$\tostrat$-evaluation $\deriv: \tm \tostrat^* \decode\statetwo$.

\item \emph{Evaluations to runs}: for every $\tostrat$-evaluation $\deriv: \tm \tostrat^* \tmtwo$ there exists a 
$\mach$-run $\exec: \compilrel\tm\state \tomachine^* \statetwo$ such that $\decode\statetwo = \tmtwo$.

\item \emph{$\beta$-Matching}: in both previous points the number $\sizebeta\exec$ of $\beta$-transitions in $\exec$ is 
exactly the length $\size\deriv$ of the evaluation $\deriv$, \ie $\size\deriv = \sizebeta\exec$.
\end{enumerate}
\end{definition}
Next, we give sufficient conditions that a machine and a strategy have to satisfy in order for the former to implement the latter, what we call \emph{an implementation system}. In the literature, strategies and machines are usually assumed to be deterministic. In Accattoli et al. \cite{DBLP:conf/lics/AccattoliCC21}, there is the case of a deterministic machine implementing a diamond strategy. Here we shall have a diamond machine implementing a diamond strategy, which is why the requirements are a bit different than for previous notion of implementation systems in the literature \cite{DBLP:journals/scp/AccattoliG19,DBLP:conf/ppdp/AccattoliCGC19,DBLP:conf/lics/AccattoliCC21}.

\begin{definition}[Implementation system]
  \label{def:implementation}
  A machine $\mach=(\States, \tomach, \compilrel\cdot\cdot, \decode\cdot)$ and a strategy $\tostrat$ form an \emph{implementation system} if:
  \begin{enumerate}
		\item\label{p:def-overhead-transparency} \emph{Overhead transparency}: $\state \tomacho \statetwo$ implies $\decode\state = \decode\statetwo$;
		\item\label{p:def-beta-projection} \emph{$\beta$-projection}: $\state \tomachhole\beta \statetwo$ implies $\decode\state \tostrat \decode\statetwo$;
		\item\label{p:def-overhead-terminate}	\emph{Overhead termination}:  $\tomacho$ terminates;	
	\item\label{p:def-beta-reflection} \emph{$\beta$-reflection}:  if $\state$ is $\tomacho$-normal and $\decode\state \tostrat \tmtwo$ then there exists $\statetwo$ such that $\state \tomachb \statetwo$ and $\decode\statetwo=\tmtwo$.
  \end{enumerate}
\end{definition}

The first two properties guarantee that the \emph{runs to evaluations} part of the implementation holds, the third and fourth properties instead induce the \emph{evaluation to runs} part, which is slightly more delicate. In the deterministic case, such a second part usually follows from a weaker notion of implementation system, where $\beta$-reflection is replaced by the weaker \emph{halt property}, stating that \emph{if $\state$ is final then $\decode\state$ is normal}. The diamond aspect of our study requires the stronger $\beta$-reflection property, which actually subsumes the halt one. Indeed, if $\decode\state$ is not normal then by $\beta$-reflection $\state$ is not final.

Thanks to a simple lemma for the \emph{evaluation to runs} part (in \ifthenelse{\boolean{techreport}}{the Appendix}{the technical report \cite{accattoli2023diamond}}), we obtain the following abstract implementation theorem.

\begin{toappendix}
\begin{theorem}[Sufficient condition for implementations]
\label{thm:abs-impl}
  Let  $\mach$ be a machine and $\tostrat$ be a strategy forming an implementation system.
  Then, $\mach$ implements $\tostrat$. More precisely, $\beta$-projection and overhead transparency imply the \emph{runs to evaluations} part (plus $\beta$-matching), and overhead termination and $\beta$-reflection imply the \emph{evaluations to runs} part (plus $\beta$-matching).
\end{theorem}
\end{toappendix}

\section{Preliminaries: The Milner Abstract Machine}

The new machine for the external strategy that we are about to introduce builds on the Milner Abstract Machine (shortened to MAM) for weak head reduction by Accattoli et al. \cite{DBLP:conf/icfp/AccattoliBM14}, that is probably the simplest abstract machine for the $\l$-calculus in the literature. In this section, we overview the MAM, the data structures and transitions of which are defined in \reffig{mam}. 

\paragraph{Data Structures.} The MAM has two data structures, the stack $\stack$ and the environment $\env$, which are lists. We use '$\cons$' for consing a single element onto a list,
but also for list concatenation, so for instance $\stack \cons \stacktwo$ stands
for the concatenation of stacks.
The set of variables bound by an environment $\genv=\esub{\var_1}{\tm_1}\ldots\esub{\var_k}{\tm_k}$ is $\set{\var_1,\ldots,\var_k}$ and it is noted $\dom\genv$.

\begin{figure}[t!]
\begin{center}
  \begin{tabular}{c}
  $\begin{array}{r\colspace rll\colspace||\colspace r\colspace rclll}
  \multicolumn{8}{c}{\textsc{Data structures, States, and Initialization}}
  \\[6pt]
  \textsc{Stacks} & \stack,\stacktwo & \grameq & \stempty \mid \tm\cons\stack
  &
  \textsc{Environments} & \env,\envtwo & \grameq & \eempty \mid \esub\var\tm\cons\env
  \\
    \textsc{States} & \state,\statetwo & \grameq & \mamstate\tm\stack\env
    &
  \textsc{Initialization} & \compilrel\tm\state & \mbox{if} & \state=\mamstate{\rename\tm}\stempty\eempty
    \end{array}$
    \\
    \\
  $
    \begin{array}{|c@{\SEP}c@{\SEP}c||@{\hspace{6pt}}c@{\hspace{6pt}}||c@{\SEP}c@{\SEP}c|l|l}
      \multicolumn{7}{c}{\textsc{Transitions}}
      \\[6pt]
   \hline
  \textsc{Active Term} & \textsc{Stack} & \textsc{Env}&&  \textsc{Active Term} & \textsc{Stack} & \textsc{Env}
  \\
    \hline
    \hline
    \tm\tmtwo & \stack & \genv
    & \toexamap &
    \tm & \tmtwo\cons\stack & \genv
    \\
    \hline
    \l\var.\tm & \tmtwo\cons\stack & \genv
    & \toexambeta &
    \tm & \stack & \esub{\var}{\tmtwo}\cons\genv
  \\
  \hline
    \var & \stack & \genv
    & \toexamsub &
    \rename{\genv(\var)} & \stack & \genv 
    \\
    \multicolumn{7}{|r|}{\text{If $\var\in\dom\env$}}
   \\
   \hline
   \end{array}$
%
  \end{tabular}
  \end{center}
  \caption{Definition of the Milner Abstract Machine (MAM).}
    \label{fig:mam}
\end{figure}
\paragraph{Transitions of the MAM.} \correction{A term $\tm$ is initialized into an initial state $\compilrel\tm\state$ by simply using an arbitrary well-named renaming $\rename\tm$} as active term together with empty stack and environment. The MAM \emph{searches} for $\beta$-redexes in the active term by descending on the left of applications via transition $\toexamap$, while accumulating arguments on the (applicative) \emph{stack}, which is simply a stack of terms. If it finds an abstraction $\la\var\tm$ and the stack has $\tmtwo$ on top, then it performs the machine analogous of a $\beta$-redex, that is a $\toexambeta$ transition, which adds the entry $\esub\var\tmtwo$ on top of the \emph{environment}, to be understood as a delayed, explicit substitution. If the MAM finds a variable $\var$, then it looks up in the environment $\env$ if it finds an associated entry $\esub\var\tm$, and replaces $\var$ with an $\alpha$-renamed $\rename\tm$ copy of $\tm$. 

Transitions $\toexamap$ and $\toexamsub$ are the overhead transitions of the MAM, that is, $\tomacho\defeq \tomachhole{\seasym_@,\subsym}$, and $\toexambeta$ is its only $\beta$-transition. The MAM is deterministic. 

\paragraph{Read-Back.} The read-back of MAM states to terms can be defined in at least two ways, by either first reading back the environment or the stack. Here we give an \correction{environment-first definition}, which shall be used also for the \EXAM.

\begin{definition}[MAM read-back]
The read-back $\decenv\tm\env$ and $\decenv{\stack}{\genv}$ of terms and stack with respect to an environment $\env$ are the terms and stacks given by:
  \begin{center}$
    \begin{array}{r\colspace rcl\colspace\colspace rcl}
      \textsc{Terms} &\decenv{\tm}{\eempty} & \defeq & \tm 
      & 
      \decenv{\tm}{\esub{\var}{\tmtwo}\cons\genv} & \defeq & \decenv{(\tm\sub{\var}{\tmtwo})}{\genv}
      \\
      \textsc{Stacks} & \decenv\stempty\env & \defeq & \stempty
      & 
      \decenv{\tm\cons\stack}{\genv} & \defeq & \decenv\tm\genv \cons \decenv\stack\env
    \end{array}$
  \end{center}
    The read-back $\decst{\tm}{\stack}$ of $\tm$ with respect to a stack $\stack$ is the term given by:
  \begin{center}
    $\begin{array}{rcl\colspace\colspace\colspace rcl}
      \decst{\tm}{\stempty}          & \defeq & \tm 
      &
      \decst{\tm}{\tmtwo\cons\stack} & \defeq & \decst{(\tm\,\tmtwo)}{\stack} \\
    \end{array}$
  \end{center}
  Finally, the read-back of a state is defined as $\decode{\mamstate\tm\stack\env} \defeq \decst{\decenv{\tm}\env}{\decenv\stack\env}$.
\end{definition}


\begin{theorem}[\cite{DBLP:conf/icfp/AccattoliBM14}]
The MAM implements weak head reduction $\towh$.
\end{theorem}

Environments are defined as \emph{lists} of entries, but they are meant to be concretely implemented as a store, without a rigid list structure. The idea is that variables are implemented as memory locations, as to obtain constant-time access to the right entry of the environment via the operation $\env(\var)$. It is nonetheless standard to define environments as lists, as it helps one stating invariants concerning them. For more implementative details, see Accattoli and Barras \cite{DBLP:conf/ppdp/AccattoliB17}.

\paragraph{Comparison with the KAM.} For the reader acquainted with the famous Krivine Abstract Machine (KAM), the difference is that the stack and the environment of the MAM contain \emph{codes}, not \emph{closures} as in the KAM, and that there is a single \emph{global} environment instead of many \emph{local} environments. A global environment indeed circumvents the complex mutually recursive notions of \emph{local environment} and \emph{closure}, at the price of the explicit $\alpha$-renaming $\rename{\tm}$ which is applied \emph{on the fly} in $\toexamsub$. The price however is negligible, at least theoretically, as the asymptotic complexity of the machine is not affected, see Accattoli and co-authors \cite{DBLP:conf/icfp/AccattoliBM14,DBLP:conf/ppdp/AccattoliB17} (the same can be said of variable names vs de Bruijn indexes/levels).

\section{From the MAM to Normal Form Using Backtracking}
Here, we quickly recall the usual way in which a machine such as the MAM can be extended to implement reduction to normal form---in this case leftmost reduction---by searching for redexes into abstractions bodies and arguments, and by adding backtracking. We follow the pattern used by Accattoli et al. \cite{DBLP:conf/aplas/AccattoliBM15}, to which we point the interested reader for extended explanations.

\paragraph{Extending the MAM.} The extended MAM in \reffig{strong-MAM} has an \emph{abstraction stack} $\absstack$ collecting the list of abstractions into which the machine entered, and a \emph{dump} $\dump$ to navigate the tree structure of applications. Moreover, it has two phases: $\skeval$ for the usual search for redexes and evaluation, and $\skback$ for backtracking. Search into arguments is done by first switching to backtracking.

The next section shows how to avoid backtracking, by changing the data structures. The obtained machine shall actually be more general, as the new data structures enable more flexibility with respect to the implemented strategy.
\begin{figure}[t]
\begin{center}
\fontsize{8}{11} \selectfont
  $
    \begin{array}{|c@{\SEP}c@{\SEP}c@{\SEP}c@{\SEP}c@{\SEP}c||c||c@{\SEP}c@{\SEP}c@{\SEP}c@{\SEP}c@{\SEP}c|l|l}
   \hline
  \textsc{Ab.} & \textsc{Tm} & \textsc{St.} & \textsc{Dump} & \textsc{En.} & \textsc{Ph.}&& \textsc{Ab.} & \textsc{Tm} & \textsc{St.} & \textsc{Dump} & \textsc{En.} & \textsc{Ph.}
  \\
    \hline
    \hline
    \absstack & \tm\tmtwo & \stack & \dump & \genv & \skeval
    & \toexamap &
    \absstack & \tm & \tmtwo{\cons}\stack & \dump & \genv & \skeval
    \\
    \hline
    \absstack & \l\var.\tm & \tmtwo{\cons}\stack & \dump & \genv & \skeval
    & \toexambeta &
    \absstack & \tm & \stack & \dump & \esub{\var}{\tmtwo}{\cons}\genv & \skeval
  \\
  \hline
    \absstack & \var & \stack & \dump & \genv & \skeval
    & \toexamsub &
    \absstack & \rename{\genv(\var)} & \stack & \dump & \genv & \skeval
    \\
    \multicolumn{13}{|r|}{\text{If $\var\in\dom\env$}}
   \\
   \hline

    \absstack & \la\var\tm & \stempty & \dump & \genv & \skeval
    & \toexamlam &
    \absstack{\cons}\var & \tm & \stempty & \dump & \genv & \skeval
     \\
  	\hline
    \absstack & \var & \stack & \dump & \genv & \skeval
    & \tomachhole{\skeval\skback} &
    \absstack & \var & \stack & \dump & \genv & \skback
    \\
    \multicolumn{13}{|r|}{\text{If $\var\notin\dom\env$}}

      \\
  	\hline
    \absstack & \tm & \tmtwo{\cons}\stack & \dump & \genv & \skback
    & \tomachhole{\skback\skeval} &
    \absstack & \tmtwo & \stempty & (\absstack{,}\tm{,}\stack){\cons}\dump & \genv & \skeval

      \\
  	\hline
    \absstack{\cons}\var & \tm & \stempty & \dump & \genv & \skback
    & \tomachhole{\skback_\lambda} &
    \absstack & \la\var\tm & \stempty & \dump & \genv & \skback

      \\
  	\hline
    \stempty & \tm & \stempty & (\absstack{,}\tmtwo{,}\stack){\cons}\dump & \genv & \skback
    & \tomachhole{\skback_@} &
    \absstack & \tmtwo & \stack & \dump & \genv & \skback
    
    \\
	\hline
   \end{array}$
   \end{center}
   \caption{An extended MAM computing strong normal forms and resting on backtracking.}
      \label{fig:strong-MAM}
\end{figure}
\section{The External Abstract Machine}
In this section, we define the EXternal Abstract Machine (\EXAM), an abstract machine for the external strategy $\toext$, by using the MAM as a sort of building block. The \EXAM is given in \reffig{exam} and explained in the following paragraphs. An example of run is given at the end of this section.

\paragraph{Data Structures.} The \EXAM has three data structures, two of which are new with respect to the MAM:
\begin{itemize}
\item The \emph{approximant} (of the normal form) $\boapr$, which collects the parts of the normal form already computed by the run of the \EXAM. The approximant is a \emph{named multi-context}, defined below, that is, a context with zero, one, or more named holes $\boctxhole\name$, each one identified by a distinct name $\name$, $\nametwo$, etc. 
\item The \emph{pool} $\pool$, which is a data structure containing a set of named MAM \emph{jobs}, providing operations for scheduling the execution of these jobs. Each named job $\jobn\name$ has shape $(\tm,\stack)_\name$, that is, it contains a term and a stack. The idea is that the job $\jobn\name=(\tm,\stack)_\name$ of name $\name$ is executing the term corresponding to $(\tm,\stack)$ and that the result of such execution shall be plugged in the approximant $\boapr$, replacing the hole $\boctxhole\name$. Pools are discussed in detail below.
\item The \emph{environment} $\env$, which is as for the MAM except that it is shared among all the jobs in the pool.
\end{itemize}

\paragraph{Transitions and Functioning of the \EXAM.} \correction{A term $\tm$ is initialized into an initial state $\compilrel\tm\state$ by creating a pool with a single named job $(\rename\tm,\stempty)_\name$ (having a well-named $\rename\tm$ version of $\tm$ and an empty stack)} and pairing it with the approximant $\boapr = \boctxhole\name$ and empty environment. The \EXAM proceeds as the MAM until it reaches a MAM final state. Let us consider the normal forms for weak head reduction and the corresponding final states of the MAM, which are of two kinds:
\begin{enumerate}
\item \emph{Abstractions (with no arguments)}: the $\towh$ normal form is $\la\var\tmtwo$ which is the read-back of a final MAM state $(\la\var\tm,\stempty,\env)$ with empty stack \correction{(that is, $\tmtwo=\decenv\tm\env$)}. In this case, the \EXAM performs a $\toexamlam$ transition, storing $\la\var\boctxhole\name$ into the approximant $\boapr$ at $\name$, and adding a named job $(\tm,\stempty)_\name$ to the pool $\pool$.

\item \emph{Possibly applied variables (with no substitution)}: the $\towh$ normal form is $\var \tmtwo_1\ldots \tmtwo_n$ with $n\geq 0$, which is the read-back of a final state $(\var,\tm_1\cons\ldots\cons\tm_n,\env)$ with $\var\notin\dom\env$ \correction{(that is, $\tmtwo_i=\decenv{\tm_i}\env$)}. In this case, the \EXAM performs a $\toexamvar$ transition. If $n = 0$ then $\toexamvar$ simply adds $\var$ into the approximant $\boapr$ at $\name$. If $n>0$ then $\toexamvar$ adds $n$ new named jobs $(\tm_1,\stempty)_{\nametwo_1}, \mydots, (\tm_n,\stempty)_{\nametwo_n}$ to the pool $\pool$ and adds $\var\boctxhole{\name_1}\mydots \boctxhole{\name_1}$ into the approximant $\boapr$ at $\name$.
\end{enumerate}
Transitions $\toexamap$, $\toexamlam$, and $\toexamvar$ are the search transitions of the \EXAM. Together with $\toexamsub$, they are the overhead transitions of the \EXAM, that is, $\tomacho \defeq \tomachhole{\seasym_@,\subsym, \seasym_\l, \seavarsym} $, and $\toexambeta$ is its only $\beta$-transition. The transition relation of the \EXAM is the union of all these relations,
\ie
${\toexam} \defeq \tomachhole{\seasym_@,\beta, \subsym, \seasym_\l, \seavarsym}$.

\begin{figure}[t!]
\begin{center}
\begin{tabular}{c}
  $\begin{array}{r@{\hspace{10pt}}rcl\colspace r@{\hspace{10pt}}rcl}
    \multicolumn{8}{c}{\textsc{Data structures, States, and Initialization}}
  \\[6pt]

    \textsc{Approx.}
      & \boapr & ::= & \boctxhole\name \mid \soapr \mid \lam{\var}{\boapr}
  &
    \textsc{Rigid approx.}
      & \soapr & ::= & \var \mid \soapr\,\boapr
  \\
    \textsc{Jobs}
      & \jobn\name & ::= & \cplace{\name}{\tm}{\stack}
  &
    \textsc{States}
      & \state & ::= & \examstate{\boapr}{\pool}{\genv} \mbox{ with }\pool \mbox{ a pool}
  \\
  &&&&
    \textsc{Initialization} & \compilrel\tm\state & \mbox{if} & \state=\examstate{\boctxhole\name}{\new{(\rename\tm,\stempty)_\name}} \eempty
  \end{array}$
\\\\

  $\begin{array}{|c@{\SEP}c@{\SEP}c||@{\hspace{6pt}}c@{\hspace{6pt}}||c@{\SEP}c@{\SEP}c|}
  \multicolumn{7}{c}{\textsc{Transitions}}
  \\[6pt]
   \hline
  \textsc{Ap.} & \textsc{Pool} & \textsc{Env}&&  \textsc{Ap.} & \textsc{Pool} & \textsc{Env}
  \\
    \hline
    \hline
        \boapr
      & \cplace{\name}{\tm\,\tmtwo}{\stack}\pop\pool
      & \genv
    & \toexamap &
        \boapr
      & \cplace{\name}{\tm}{\tmtwo\cons\stack}\push\pool
      & \genv
  \\
  \hline
        \boapr
      & \cplace{\name}{\lam{\var}{\tm}}{\tmtwo\cons\stack}\pop\pool
      & \genv
    & \toexambeta &
        \boapr
      & \cplace{\name}{\tm}{\stack}\push\pool
      & \esub{\var}{\tmtwo}\cons\genv
  \\
  \hline
        \boapr
      & \cplace{\name}{\var}{\stack}\pop\pool
      & \genv
    & \toexamsub &
        \boapr
      & \cplace{\name}{\rename{\genv(\var)}}{\stack}\push\pool
      & \genv
      \\ \multicolumn{7}{|r|}{\text{If $\var\in\dom\env$ and $\rename{\tm}$ is a fresh renaming of $\tm$}}
  \\
  \hline
        \boapr
      & \cplace{\name}{\lam{\var}{\tm}}{\stempty}\pop\pool
      & \genv
    & \toexamlam &
        \boaprtwo
      & \cplace{\name}{\tm}{\stempty}\push\pool
      & \genv
      \\ \multicolumn{7}{|r|}{\text{With $\boaprtwo \defeq \boapr\boctxholep\name{\lam{\var}{\boctxhole\name}}$}}
  \\
  \hline
        \boapr
      & \cplace{\name}{\var}{\tm_1\cons\mydots\cons\tm_n}\pop\pool
      & \genv
    & \toexamvar &
        \boaprtwo
      & \cplace{\nametwo_1}{\tm_1}{\stempty}\cons\mydots\cons\cplace{\nametwo_n}{\tm_n}{\stempty} \add \pool
      & \genv
      \\ \multicolumn{7}{|r|}{\text{If $\var \not\in \dom\genv$, and with $n \geq 0$, $\boaprtwo \defeq \boapr\boctxholep\name{\var\,\boctxhole{\nametwo_1}\mydots\boctxhole{\nametwo_n}}$, and $\nametwo_1,\mydots,\nametwo_n$ fresh }}
      \\
      \hline
  \end{array}
  $
  \end{tabular}
\end{center}
  \caption{Definition of the EXternal Abstract Machine (\EXAM).}
    \label{fig:exam}
\end{figure}

\paragraph{Pool Interface and Templates.} The \EXAM selects at each step a (named) job from the pool---the one performing the step---according to a possibly non-deter\-ministic policy, and drops it back in the pool after the transition, unless the job is over, which happens in transition $\toexamvar$. In general, \emph{dropping a job back into a pool} and \emph{adding a job to a pool} are not the same operation, since the former applies to jobs that were in the pool before being selected, while addition is reserved to new jobs. We actually abstract away from a job scheduling policy and from the details of the pool data structure: the pool is an abstract \emph{interface} which can be realized by various concrete data structures called \emph{pool templates}.
\begin{definition}[Pool templates]
A pool template is a data structure $\pool$ coming with the following five operations of the \emph{pool interface}:\label{def:pools}
\begin{itemize}
\item \emph{Names, support, and new}: \correction{there are a name function $\names\pool= \set{\name_1,\mydots,\name_n}$ providing the finite and possibly empty set of the names of the jobs in the pool ($\nat\ni n\geq 0$), a support function $\supp\pool= \set{\jobn{\name_1},\mydots,\jobn{\name_n} }$ providing the set of jobs in the pool (indexed by $\names\pool$)}, and a function $\new{\jobn\name}$ creating a pool containing $\jobn\name$, that is, such that $\supp{\new{\jobn\name}} = \set{\jobn\name}$. 

\item \emph{Selection}: there is a selection relation $\pop (\pool, \jobn\name, \pooltwo)$ such that $\jobn\name \in \supp\pool$ and $\supp\pooltwo = \supp\pool\setminus\set{\jobn\name}$. The intuition is that $\pooltwo$ is $\pool$ without $\jobn\name$, which has been selected and removed from $\pool$. \correction{There is a \emph{choice constraint}: if $\pool$ is non-empty then $\pop (\pool, \jobn\name, \pooltwo)$ holds for some $\jobn\name$ and $\pooltwo$.} We write $\jobn\name \pop \pooltwo$ for a pool $\pool$ such that $\pop(\pool, \jobn\name, \pooltwo)$.

\item \emph{Dropping}: there is a dropping function $\push(\jobn\name, \pool)=\pooltwo$ defined when $\name \notin \names\pool$ and such that $\supp\pooltwo = \supp\pool\cup\set{\jobn\name}$. Dropping is meant to add a job $\jobn\name$ back to a pool $\pool$ from which $\jobn\name$ was previously selected. It is not necessarily the inverse of selection. We write $\jobn\name \drop \pool$ for the pool $\push(\jobn\name, \pool)$.

\item \emph{Adding}: similarly, there is an adding function $\add(\jobn\name, \pool)=\pooltwo$ defined when $\name \notin \names\pool$ and such that $\supp\pooltwo = \supp\pool\cup\set{\jobn\name}$. Adding is meant to add a new job $\jobn\name$ to a pool $\pool$, that is, a job that has never been in $\pool$. We write $\jobn\name \add \pool$ for $\add(\jobn\name, \pool)=\pooltwo$, and extend it to lists as follows:  $\eempty \add \pool \defeq \pool$,  and $\jobn{\name_1}\cons\mydots\cons\jobn{\name_n} \add \pool \defeq \jobn{\name_n} \add \left(\jobn{\name_1}\cons\mydots\cons\jobn{\name_{n-1}} \add \pool\right)$.
\end{itemize}
\end{definition}

\paragraph{\SEXAM.} The simplest pool template is the \emph{set template} where pools $\pool$ are sets of named jobs, the support is the pool itself (and the name set is as expected), $\new{\jobn{\name}}$ creates a singleton with $\jobn\name$, selection is the relation $\set{(\pool, \jobn\name, \pool\setminus\set{\jobn\name})\,|\, \jobn\name \in \pool}$, and both dropping and adding are the addition of an element. The set template models the most general behavior, as \emph{any} job of the pool can then be selected for the next step. The \EXAM instantiated on the set template is called \emph{\SEXAM}. Other templates shall be considered at the end of the paper, motivating in particular the distinction between dropping and adding.


\paragraph{Approximants and Named Multi-Contexts.} The definition of the \EXAM rests on approximants, which are stable prefixes of normal forms, that is, normal forms from which some sub-terms have been removed and replaced with named holes. In fact, we are going to introduce more general \emph{(named) multi-contexts} to give a status to approximants in which some but not all holes have been replaced by an \emph{arbitrary term}---which shall be needed in proofs (when manipulating the read-back)---thus losing their "normal prefix" property.

\begin{definition}[Named multi-contexts]
A \emph{(named) multi-context} $\ntm$ is a $\lambda$-term in which there may appear
free occurrences of (named) holes, \ie:
\begin{center}$
  \textsc{(Named) Multi-contexts}\ \ \ \ntm ::= \var \mid \boctxhole\name \mid \lam{\var}{\ntm} \mid \ntm\,\ntm
  $\end{center}
The plugging $\ntm\tsub{\name}{\ntmtwo}$ of $\name$ by $\ntmtwo$ in $\ntm$,
is the capture-allowing substitution of $\boctxhole\name$ by $\ntmtwo$ in $\ntm$.
We write $\fn{\ntm}$ for the set of names that occur in $\ntm$.  \correction{We shall use only multi-contexts where named holes have pairwise \emph{distinct names}}.
\end{definition}

Note that a multi-context $\ntm$ without holes is simply a term, thus the defined notion of plugging subsumes the plugging $\ntm\tsub{\name}{\tm}$ of terms in multi-contexts.

Approximants $\boapr$ are defined in \reffig{exam} by mutual induction with rigid approximants $\soapr$, and are special cases of multi-contexts. Alternative streamlined definitions for (rigid) approximants are (possibly $\vartwo = \var_i$ for some $i\in\set{1,\ldots,n}$):
\begin{center}$
  \begin{array}{rll\colspace \colspace \colspace rllllllll}
    \soapr & ::=  & \var\,\boapr_1\mydots\boapr_n
  &
    \boapr & ::=  & \lam{\var_1\mydots\var_n}{\ctxhole_\name} 
           & \mid & \lam{\var_1\mydots\var_n}{\vartwo\,\boapr_1\mydots\boapr_n}
  \end{array}
$\end{center}
Note that in $\boapr$ and $\soapr$ holes are never applied, that is, they are \emph{non-applying} multi-contexts. For the sake of readability, in the paper we only give statements about approximants, which are then reformulated in \ifthenelse{\boolean{techreport}}{the Appendix}{the technical report \cite{accattoli2023diamond}} by pairing them with a similar statement about rigid approximants, and proving the two of them simultaneously by mutual induction.

We prove two properties of approximants. Firstly, to justify that transitions $\toexamlam$ and $\toexamvar$ are well-defined,
we show that the first component of the state on their right-hand side
is indeed an approximant. Secondly, we relate approximants with normal forms, to justify the terminology.

\begin{toappendix}
\begin{lemma}[Inner extension of approximants]
If $\boapr$ is an  approximant and $\nametwo_1,\ldots,\nametwo_n\notin\fn\boapr$
then $\boapr\boctxholep\name{\lam{\var}{\boctxhole\name}}$
and $\boapr\boctxholep\name{\var\,\boctxhole{\nametwo_1}\mydots\boctxhole{\nametwo_n}}$ are
approximants.
\label{l:inn-extension-approx}
\end{lemma}
\end{toappendix}

\begin{toappendix}
\begin{lemma}
\label{l:approx-no-names-normal}
An approximant $\boapr$ without named holes is a normal form.
\end{lemma}
\end{toappendix}

\paragraph{Read-Back.} To give a notion of read-back that is independent of the pool template, we define the read-back using a set $\suppset$ of uniquely named jobs---standing for the support $\supp\pool$ of the pool---rather than the pool $\pool$ itself. Moreover, we need a way of applying the substitution induced by an environment to named jobs and sets of named jobs, which is based on the notions $\decenv\tm\env$ and $\decenv\stack\env$ for terms and stacks given for the MAM, from which we also borrow the definition of $\decst\tm\stack$.
  
  \begin{definition}[\EXAM read-back]
  Applying an environment $\env$ to jobs and job sets is defined as follows:
  \begin{center}$
    \begin{array}{r\colspace rcl\colspace\colspace rcl}
      \textsc{Jobs/jobs sets} & 
      \decenv{\cplace\name\tm\stack}{\genv} & \defeq & \cplace\name{\decenv\tm\env}{\decenv\stack\env}
      &
      \decenv{\set{\jobn{\name_1},\mydots,\jobn{\name_n}}}\genv & \defeq & \set{\decenv{\jobn{\name_1}}\genv,\mydots,\decenv{\jobn{\name_n}}\genv}
    \end{array}$
  \end{center}
The read-back of jobs, and of a multi context $\ntm$ with respect to a set of uniquely named jobs $\set{\jobn{\name_1}, \mydots, \jobn{\name_n}}$ are defined as follows:
  \begin{center}$
    \begin{array}{rcl\colspace\colspace\colspace\colspace rcl}
      \decode{\cplace\name\tm\stack} & \defeq & \decst\tm\stack
      &
      \decpool{\ntm}{\set{\jobn{\name_1}, \mydots, \jobn{\name_n}}}  & \defeq &  \ntm\tsub{\name_1}{\decode{\jobn{\name_1}}}\mydots\tsub{\name_n}{\decode{\jobn{\name_n}}} 
    \end{array}$
  \end{center}
An \EXAM state $\state$ is read-back as a multi-context setting $\decode{\examstate{\boapr}{\pool}{\genv}} \defeq \decpool{\boapr}{\decenv{\supp\pool}\env}$.

\end{definition}

\paragraph{Diamond.} Since the selection operation is non-deterministic, the \EXAM in general is non-deterministic. The most general case is given by the \emph{\SEXAM}, which is the \EXAM instantiated with the set template for pools described after \refdef{pools}. As for the external strategy, the \SEXAM has the diamond property up to a slight glitch: swapping the order of two $\beta$-transitions on two different jobs, adds entries to the environment in different orders.

Let $\approx$ be the minimal equivalence relation on environments containing the following relation:
\begin{center}
$\env\cons \esub\var\tm \cons \esub\vartwo\tmtwo \cons\envtwo\ \ \sim\ \ \env\cons \esub\vartwo\tmtwo\cons \esub\var\tm  \cons\envtwo\ \ \ \ \mbox{if }\var\notin\tmtwo\mbox{ and }\vartwo\notin\tm$
\end{center}
Let $\streq$ be the relation over states s. t. $\examstate\boapr\pool{\env_1} \streq \examstate\boapr\pool{\env_2}$ if $\env_1 \approx \env_2$.

\begin{toappendix}
\begin{proposition}
\label{prop:exam-diamond}
The \SEXAM is diamond up to $\streq$, i.e., if $\state \toexam \state_1$ and $\state \toexam \state_2$ then  $\exists \statetwo_1$ and $\statetwo_2$ such that $\state_1 \toexam \statetwo_1$, $\state_2 \toexam \statetwo_2$, and $\statetwo_1 \streq \statetwo_2$.
\end{proposition}
\end{toappendix}

\correction{
\begin{example} 
\label{ex:set-exam}
The following is a possible run of the \SEXAM---that is, the \EXAM with the set template for pools---on the term $\tm \defeq \var (\Id_\vartwo \varthree) (\delta_\varfour \varthree)$ where $\Id_\vartwo=\la\vartwo\vartwo$ and $\delta_\varfour = \la\varfour\varfour\varfour$, ending in a final state.
\begin{center}
\begin{tabular}{c}
$\begin{array}{|c@{\SEP}c@{\SEP}c||@{\hspace{6pt}}c@{\hspace{6pt}}|c@{\hspace{6pt}}clll}
\cline{1-4}
	\textsc{Approx.} & \textsc{Pool} & \textsc{Env}&\textsc{Tran.}&& \textsc{Selected Job}
	\\
    \hline
    \hline
    \boctxhole\name
    & \set{\cplace{\name}{\var (\Id_\vartwo \varthree) (\delta_\varfour \varthree)}{\stempty}}
    & \eempty
    & \toexamap 
    && \name
  	\\  
    \boctxhole\name
    & \set{\cplace{\name}{\var (\Id_\vartwo \varthree)}{\delta_\varfour \varthree}}
    & \eempty
    & \toexamap
    && \name
    \\
    \boctxhole\name
    & \set{\cplace{\name}{\var}{\Id_\vartwo \varthree \cons \delta_\varfour \varthree}}
    & \eempty
    & \toexamvar
    && \name
    \\
    \var\boctxhole\nametwo\boctxhole\namethree
    & \set{\cplace{\nametwo}{\Id_\vartwo \varthree}{\stempty}, \cplace{\namethree}{ \delta_\varfour \varthree}{\stempty}}
    & \eempty
    & \toexamap
    && \namethree
    \\
    \var\boctxhole\nametwo\boctxhole\namethree
    & \set{\cplace{\nametwo}{\Id_\vartwo \varthree}{\stempty}, \cplace{\namethree}{ \delta_\varfour }{\varthree}}
    & \eempty
    & \toexambeta
    && \namethree
    \\
    \var\boctxhole\nametwo\boctxhole\namethree
    & \set{\cplace{\nametwo}{\Id_\vartwo \varthree}{\stempty}, \cplace{\namethree}{ \varfour\varfour }{\stempty}}
    & \esub\varfour\varthree
    & \toexamap
    && \nametwo
    \\
    \var\boctxhole\nametwo\boctxhole\namethree
    & \set{\cplace{\nametwo}{\Id_\vartwo }{\varthree}, \cplace{\namethree}{ \varfour\varfour }{\stempty}}
    & \esub\varfour\varthree
    & \toexamap
    && \namethree
    \\
    \var\boctxhole\nametwo\boctxhole\namethree
    & \set{\cplace{\nametwo}{\Id_\vartwo }{\varthree}, \cplace{\namethree}{ \varfour }{\varfour}}
    & \esub\varfour\varthree
    & \toexambeta
    && \nametwo
    \\
    \var\boctxhole\nametwo\boctxhole\namethree
    & \set{\cplace{\nametwo}{\vartwo }{\stempty}, \cplace{\namethree}{ \varfour }{\varfour}}
    & \esub\vartwo\varthree \cons \esub\varfour\varthree
    & \toexamsub
    && \nametwo
    \\
    \var\boctxhole\nametwo\boctxhole\namethree
    & \set{\cplace{\nametwo}{\varthree }{\stempty}, \cplace{\namethree}{ \varfour }{\varfour}}
    & \esub\vartwo\varthree \cons \esub\varfour\varthree
    & \toexamsub
    && \namethree
    \\
    \var\boctxhole\nametwo\boctxhole\namethree
    & \set{\cplace{\nametwo}{\varthree }{\stempty}, \cplace{\namethree}{ \varthree }{\varfour}}
    & \esub\vartwo\varthree \cons \esub\varfour\varthree
    & \toexamvar
    && \namethree
    \\
    \var\boctxhole\nametwo(\var\boctxhole{\namethree'})
    & \set{\cplace{\nametwo}{\varthree }{\stempty}, \cplace{\namethree'}{ \varfour }{\stempty}}
    & \esub\vartwo\varthree \cons \esub\varfour\varthree
    & \toexamsub
    && \namethree'
    \\
    \var\boctxhole\nametwo(\var\boctxhole{\namethree'})
    & \set{\cplace{\nametwo}{\varthree }{\stempty}, \cplace{\namethree'}{ \varthree }{\stempty}}
    & \esub\vartwo\varthree \cons \esub\varfour\varthree
    & \toexamvar
    && \nametwo
    \\
    \var\varthree(\var\boctxhole{\namethree'})
    & \set{\cplace{\namethree'}{ \varthree }{\stempty}}
    & \esub\vartwo\varthree \cons \esub\varfour\varthree
    & \toexamvar
    && \namethree'
   
    \\
    \var\varthree(\varthree\varthree)
    & \emptyset
    & \esub\vartwo\varthree \cons \esub\varfour\varthree
    & 
    &&
  	\\
	\cline{1-4}
  \end{array}
  $
  \end{tabular}
\end{center}

\end{example}
}

\section{Runs to Evaluations}
In this section, we develop the projection of \EXAM runs on external evaluations, and then instantiates it with a deterministic pool template obtaining runs corresponding to leftmost evaluations.

\paragraph{Overhead Transparency.} By the abstract recipe for implementation theorems in \refsect{prel-abs-mach}, to project runs on evaluations we need to prove \emph{overhead transparency} and \emph{$\beta$-projection}. Overhead transparency is simple, it follows from the definition of read-back plus some of its basic properties (in \ifthenelse{\boolean{techreport}}{the Appendix.}{the technical report \cite{accattoli2023diamond}}).

\begin{toappendix}
\begin{proposition}[Overhead transparency]
\label{prop:overhead-transparency}If $\state \tomacho \state'$ then $\decode{\state} = \decode{\state'}$.
\end{proposition}
\end{toappendix}

\paragraph{Invariants.} To prove the $\beta$-projection property, we need some invariants of the \EXAM. Firstly, we deal with a set of invariants  concerning variable names, hole names, and binders. A notable point is that they are proved by using only properties of the pool interface, and are thus valid for every pool template instantiation of the \EXAM.

\emph{Terminology}: a {\em binding occurrence} of a variable $\var$ is an occurrence of $\lam{\var}{\tm}$ in $\boapr$, $\pool$ or $\genv$, or an occurrence of $\esub{\var}{\tm}$ in $\genv$, for some $\tm$.

\begin{toappendix}
\begin{lemma}[\EXAM Invariants]
Let $\state = \examstate{\boapr}{\pool}{\genv}$ be an \EXAM reachable state
reachable.
 Then: \label{l:exam_invariant}
\begin{enumerate}
\item 
\emph{Uniqueness}. There are no repeated names in $\boapr$.

\item 
  \emph{Freshness.}
  Different binding occurrences in $\state$ refer to different variable names.

\item 
\emph{Bijection}.  The set of names in $\boapr$ is in 1--1 correspondence with the set of names in $\pool$, that is, $\fn{\boapr} = \names\pool$.

\item 
\emph{Freeness}.  The free variables of $\boapr$ are globally free, that is, $\fv\boapr\cap \dom\genv = \emptyset$.

\item 
  \emph{Local scope.}
  For every sub-term of the form $\lam{\var}{\tm}$
  in \correction{a job in $\supp\pool$} or in $\env$,
  there are no occurrences of $\var$ outside of $\tm$.
  Moreover, in the environment $\esub{\var_1}{\tm_1}\cons\mydots\cons\esub{\var_n}{\tm_n}$,
  there are no occurrences of $\var_i$ in $\tm_j$ if $i \leq j$.
\end{enumerate}
\end{lemma}
\end{toappendix}

The read-back of an \EXAM state is defined as a multi-context, but for reachable terms it is in fact a term, as stated by the second point of the next lemma, which is proved by putting together the bijection invariant for reachable states  (\reflemma{exam_invariant}.3) and the first point.
\begin{toappendix}
\begin{lemma}[Reachable states read back to terms]
\label{l:pure_decoding}
\begin{enumerate}
\item Let $\boapr$ be an  approximant and let $\suppset$ be a set of uniquely named jobs
such that $\fn{\boapr} \subseteq \names\suppset$.
Then $\decpool{\boapr}{\suppset}$ is a term.

\item Let $\state$ be a reachable state. Then its read-back $\decode\state$ is a term.
\end{enumerate}
\end{lemma}
\end{toappendix}

\paragraph{Contextual Read-Back.} The key point of the $\beta$-projection property is proving that the read-back of the data structures of a reachable state without the active term/job provides an evaluation context---an external context in our case. This is ensured by the following lemma. It has a simple proof (using \reflemma{pure_decoding}) because we can state it about approximants without mentioning reachable state, given that we know that the first component of a reachable state is always an approximant (because of \reflemma{inn-extension-approx}). The lemma is then used in the proof of the $\beta$-projection property.

\begin{toappendix}
\begin{lemma}[External context read-back]
\label{l:almost_pure_decoding}
Let $\suppset$ be a set of uniquely named jobs and $\boapr$ be an approximant with no repeated names
such that $\fn{\boapr} \setminus \names\suppset = \set{\name}$.
Then $\decpool{\boapr}{\suppset}\tsub{\name}{\ctxhole}$  is an external context.
\end{lemma}
\end{toappendix}


\begin{toappendix}
\begin{theorem}[$\beta$-projection]
If $\state \toexambeta \state'$ then $\decode{\state} \toext \decode{\state'}$.
\label{thm:exam-beta-projection}
\end{theorem}
\end{toappendix}

Now, we obtain the \emph{runs to evaluations} part of the implementation theorem, which by the theorem about sufficient conditions for implementations (\refthm{abs-impl}) follows from overhead transparency and $\beta$-projection.

\begin{toappendix}
\begin{corollary}[\EXAM runs to external evaluations]
For any \EXAM run $\exec: \compilrel\tm\state \toexam^* \statetwo$ there exists a 
$\toext$-evaluation $\deriv: \tm \toext^* \decode\statetwo$. Moreover, $\size\deriv = \sizebeta\exec$.
\label{coro:exam-correctness}
\end{corollary}
\end{toappendix}

Last, we analyze final states.

\begin{toappendix}
\begin{proposition}[Characterization of final states]
Let $\state$ be a reachable final state. Then there is a normal form $\nf$ such that $\state = \examstate\nf\pempty\genv$  and $\decode\state=\nf$. Moreover, if $\state \streq\statetwo$ then $\statetwo$ is final and $\decode\statetwo=\nf$.
\label{prop:halt}
\end{proposition}
\end{toappendix}

\subsection{Leftmost Runs to Leftmost Evaluations}

Now, we instantiate the \EXAM with the stack template for pools, obtaining a machine implementing leftmost evaluation.

\paragraph{\LEXAM.} Let the \emph{\LEXAM} be the deterministic variant of the \EXAM adopting the stack template for pools, that is, such that:
\begin{itemize}
\item Pools are lists $\job_{\name_1}\cons\mydots\cons \job_{\name_n}$ of named jobs, $\new{\job_\name}$ creates the list containing only $\job_\name$, and the support of a pool is the set of jobs in the list;
\item Selection pops from the pool, that is, if $\pool = \job_{\name_1}\cons\mydots\cons \job_{\name_n}$ then $\place_{\name}\pop\pool$ pops $\place_\name$ from the list $\place_\name\cons\job_{\name_1}\cons\mydots\cons \job_{\name_n}$;
\item Both dropping and adding push on the list, and are inverses of selection.
\end{itemize}

\correction{
\begin{example} 
\label{ex:leftmost-exam}
The \LEXAM run on the same term $\tm \defeq \var (\Id_\vartwo \varthree) (\delta_\varfour \varthree)$ used for the \SEXAM in \refex{set-exam} follows (excluding the first three transitions, that are the same for both machines, as they are actually steps of the MAM).
\begin{center}
\begin{tabular}{c}
$\begin{array}{|c@{\SEP}c@{\SEP}c||@{\hspace{6pt}}c@{\hspace{6pt}}|c@{\hspace{6pt}}clll}
\cline{1-4}
	\textsc{Approx.} & \textsc{Pool} & \textsc{Env}&\textsc{Trans.}&& \textsc{Selected Job}
	\\
    \hline
    \hline
    \var\boctxhole\nametwo\boctxhole\namethree
    & \cplace{\nametwo}{\Id_\vartwo \varthree}{\stempty} \cons \cplace{\namethree}{ \delta_\varfour \varthree}{\stempty}
    & \eempty
    & \toexamap
    && \nametwo
    \\
    \var\boctxhole\nametwo\boctxhole\namethree
    & \cplace{\nametwo}{\Id_\vartwo }{\varthree} \cons \cplace{\namethree}{ \delta_\varfour \varthree}{\stempty}
    & \eempty
    & \toexambeta
    && \nametwo
    \\
    \var\boctxhole\nametwo\boctxhole\namethree
    & \cplace{\nametwo}{\vartwo }{\stempty} \cons \cplace{\namethree}{ \delta_\varfour \varthree}{\stempty}
    & \esub\vartwo\varthree
    & \toexamsub
    && \nametwo
    \\
    \var\boctxhole\nametwo\boctxhole\namethree
    & \cplace{\nametwo}{\varthree}{\stempty} \cons \cplace{\namethree}{ \delta_\varfour \varthree}{\stempty}
    & \esub\vartwo\varthree
    & \toexamvar
    && \nametwo
    \\
    \var\varthree\boctxhole\namethree
    & \cplace{\namethree}{ \delta_\varfour \varthree}{\stempty}
    & \esub\vartwo\varthree
    & \toexamap
    && \namethree
    \\
    \var\varthree\boctxhole\namethree
    & \cplace{\namethree}{ \delta_\varfour}{ \varthree}
    & \esub\vartwo\varthree
    & \toexambeta
    && \namethree
    \\
    \var\varthree\boctxhole\namethree
    & \cplace{\namethree}{ \varfour\varfour}{ \stempty}
    & \esub\varfour\varthree \cons \esub\vartwo\varthree
    & \toexamap
    && \namethree
    \\
    \var\varthree\boctxhole\namethree
    & \cplace{\namethree}{ \varfour}{ \varfour}
    & \esub\varfour\varthree \cons \esub\vartwo\varthree
    & \toexamsub
    && \namethree
    \\
    \var\varthree\boctxhole\namethree
    & \cplace{\namethree}{ \varthree}{ \varfour}
    & \esub\varfour\varthree \cons \esub\vartwo\varthree
    & \toexamvar
    && \namethree
    \\
    \var\varthree(\varthree\boctxhole{\namethree'})
    & \cplace{\namethree'}{ \varfour}{ \stempty}
    & \esub\varfour\varthree \cons \esub\vartwo\varthree
    & \toexamsub
    && \namethree'
    \\
    \var\varthree(\varthree\boctxhole{\namethree'})
    & \cplace{\namethree'}{ \varthree}{ \stempty}
    & \esub\varfour\varthree \cons \esub\vartwo\varthree
    & \toexamvar
    && \namethree'
    \\
    \var\varthree(\varthree\varthree)
    & \stempty
    & \esub\varfour\varthree \cons \esub\vartwo\varthree
    & 
    && 
    \\
	\cline{1-4}
  \end{array}
  $
  \end{tabular}
\end{center}

\end{example}
}
Proving that \LEXAM runs read back to leftmost evaluations requires a new $\beta$-projection property. Its proof is based on the following quite complex invariant about instantiations of the approximants computed by the \LEXAM. \emph{Terminology}: a context $\ctx$ is \emph{non-applying} if $\ctx \neq \ctxtwop{\ctxhole\tm}$ for all $\ctxtwo$ and $\tm$, that is, it does not apply the hole to an argument.

\begin{toappendix}
\begin{proposition}[Leftmost context invariant]
Let \label{prop:lexam-complex-invariant}
\begin{itemize}
\item $\state=\examstate\boapr{\job_{\name_1}\cons\mydots\cons \job_{\name_n}}\genv$ be a reachable \LEXAM state with $n\geq 1$; 
\item $\nf_j$ be a normal form for all $j$ such that $1\leq j<n$,
\item $\tm_j$ be a term for all $j$ such that $1< j\leq n$, and 
\item $\lctx$ be a non-applying leftmost context. 
\end{itemize}
Then $\ctx_{\nf_1,\mydots,\nf_{i-1}\mid \lctx\mid \tm_{i+1},\mydots,\tm_n}^\state \defeq \boapr\tsub{\name_1}{\nf_1}\mydots\tsub{\name_{i-1}}{\nf_{i-1}}\tsub{\name_i}{\lctx} \tsub{\name_{i+1}}{\tm_{i+1}}\mydots\tsub{\name_{n}}{\tm_n}$
 is a non-applying leftmost context for every $i\in\set{1,\ldots, n}$.
\end{proposition}
\end{toappendix}

From the invariant, it follows that a reachable state less the first job reads back to a leftmost context, which implies the leftmost variant of $\beta$-projection, in turn allowing us to project \LEXAM runs on leftmost evaluations.
\begin{toappendix}
\begin{lemma}[Leftmost context read-back]
\label{l:leftmost-context-decoding}
Let $\state=\examstate\boapr{\job_{\name_1}\cons\mydots\cons\job_{\name_n}}\genv$ be a reachable \LEXAM state with $n\geq 1$ and $\state_\bullet \defeq \examstate\boapr{\job_{\name_2}\cons\mydots\cons\job_{\name_n}}\genv$. Then $\decode{\state_\bullet}\tsub{\name_1}{\ctxhole}$ is a non-applying leftmost context.
\end{lemma}
\end{toappendix}

\begin{toappendix}
\begin{proposition}[Leftmost $\beta$-projection]
Let $\state$ be a reachable \LEXAM state. If $\state \tomachb \statetwo$ then $\decode\state \tolo \decode\statetwo$.\label{prop:leftmost-beta-proj}
\end{proposition}
\end{toappendix}

\begin{toappendix}
\begin{corollary}[\LEXAM runs to leftmost evaluations]
\label{coro:lexam-correctness}
For any \LEXAM run $\exec: \compilrel\tm\state \toexam^* \statetwo$ there exists a 
$\tolo$-evaluation $\deriv: \tm \tolo^* \decode\statetwo$. Moreover, $\size\deriv = \sizebeta\exec$.
\end{corollary}
\end{toappendix}

\section{Evaluations to Runs}
Here, we develop the reflection of external evaluations to \EXAM runs. By the abstract recipe for implementation theorems in \refsect{prel-abs-mach}, one needs to prove \emph{overhead termination} and \emph{$\beta$-reflection}. 
At the level of non-determinism, the external strategy is matched by the most permissive scheduling of jobs, that is, the set template. Therefore, we shall prove the result with respect to the \SEXAM.

\paragraph{Overhead Termination.} To prove overhead termination, we define a measure. The measure does not depend on job names nor bound variable names, which is why the definition of the measure replaces them with underscores (and it is well defined even if it uses the renaming $\rename{\genv(\var)}$).

\begin{definition}[Overhead measure]
Let $\job_\name$ be a job and $\genv$ be an environment satisfying the freshness name property (together) of \reflemma{exam_invariant}, and $\state$ be a reachable state. The overhead measures $\omeas{\place_\name, \genv}$ and $\omeas{\state}$ are defined as follows:
\begin{center}
$\begin{array}{rclrcc}
\omeas{\cplace\_{\la\_\tm}{\tmtwo\cons\stack}, \genv} & \defeq & 0
\\
\omeas{\cplace\_{\la\_\tm}{\estack}, \genv} & \defeq & 1+ \omeas{\cplace\_{\tm}{\estack}, \genv}
\\
\omeas{\cplace\_{\tm\tmtwo}{\stack}, \genv} & \defeq & 1+ \omeas{\cplace\_{\tm}{\tmtwo\cons\stack}, \genv}
\\
\omeas{\cplace\_{\var}{\estack}, \genv} & \defeq & 1+ \omeas{\cplace\_{\rename{\genv(\var)}}{\estack}, \genv} & \mbox{if }\var\in\dom\genv
\\
\omeas{\cplace\_{\var}{\tm_1\cons\mydots\cons\tm_n}, \genv} & \defeq & 1+ \sum_{i=1}^n\omeas{\cplace\_{\tm_i}{\estack}, \genv} & \mbox{with }n\geq 0\mbox{, if }\var\notin\dom\genv
\\[5pt]
\omeas{\examstate{\boapr}{\pool}{\genv}} & \defeq & \sum_{\place_\name\in\supp\pool}\omeas{\place_\name,\genv}

\end{array}$
\end{center}
\end{definition}

\begin{toappendix}
\begin{proposition}[Overhead termination]
Let $\state$ be a \SEXAM reachable state. Then $\state\tomacho^{\omeas\state}\statetwo$ with $\statetwo$ $\tomacho$-normal.
\label{prop:overhead-termination}
\end{proposition}
\end{toappendix}

\paragraph{Addresses and $\beta$-Reflection.}  For the $\beta$-reflection property, we need a way to connect external redexes on terms with $\beta$-transitions on states. We use \emph{addresses}.

\begin{definition}[Address and sub-term at an address]
An \emph{address} $\addr$ is a string over the alphabet $\set{l,r,\l}$. The sub-term
	$\staddr\tm\addr$ of a term $\tm$ at address $\addr$ is the following partial function (the last case of which means that in any case not covered by the previous ones $\staddr\tm\addr$ is undefined):
	\[\begin{array}{rcl@{\hspace{.9cm}}rcl@{\hspace{.9cm}}rcl@{\hspace{.9cm}}rcllll}
	\staddr\tm\ems & \defeq&  \tm  
	&
	\staddr{(\tm\tmtwo)}{l\cons\addr} & \defeq &\staddr\tm\addr	
	&
	\staddr{(\tm\tmtwo)}{r\cons\addr} & \defeq &\staddr\tmtwo\addr
	&
	\staddr{(\la\var\tm)}{\l\cons\addr} & \defeq &\staddr\tm\addr 
	\\[3pt]
	\staddr\_{c\cons\addr} & \defeq&   \multicolumn{4}{l}{\bot\ \ \mbox{if }c\in\set{l,r,\l}}
	\end{array}\]
 The sub-term $\staddr\ntm\addr$ at $\addr$ of a multi-context is defined analogously. 	\correction{An address $\addr$ is defined in $\tm$ (resp. $\ntm$) if $\staddr\tm\addr\neq\bot$ (resp. $\staddr\ntm\addr\neq\bot$), and undefined otherwise.}
\end{definition}

There is a strong relationship between addresses in the approximant of a state and in the read-back of the state, as expressed by the following lemma. The lemma is then used to prove $\beta$-reflection, from which the \emph{evaluation to runs} part of the implementation theorem follows.
\begin{toappendix}
\begin{lemma}
\label{l:addresses-unfolding}
Let $\state=\examstate\boapr\pool\genv$ be a state and $\addr$ a defined address in $\boapr$. Then $\addr$ is a defined address in $\decode\state$, and $\staddr{\decode\state}\addr$ starts with the same constructor of $\staddr\boapr\addr$ unless $\staddr\boapr\addr$ is a named hole.
\end{lemma}
\end{toappendix}

\begin{toappendix}
\begin{proposition}[$\beta$-reflection]
Let $\state$ be a $\tomacho$-normal reachable state. If  $\decode\state \toext \tmtwo$ then there exists $\statetwo$ such that $\state \tomachb \statetwo$ and $\decode\statetwo=\tmtwo$.
\label{prop:beta-reflection}
\end{proposition}
\end{toappendix}

\begin{toappendix}
\begin{corollary}[Evaluations to runs]
\label{coro:exam-completeness}
For every $\toext$-evaluation $\deriv: \tm \toext^* \tmtwo$ there exists a 
 \SEXAM run $\exec: \compilrel\tm\state \toexam^* \statetwo$ such that $\decode\statetwo = \tmtwo$.
\end{corollary}
\end{toappendix}

A similar result for leftmost evaluation and the \LEXAM follows more easily from the characterization of final states (\refprop{halt}), overhead termination (\refprop{overhead-termination}), and determinism of the \LEXAM---this is the standard pattern for deterministic strategies and machines, used for instance by Accattoli et al. for their machine for leftmost evaluation \cite{DBLP:conf/aplas/AccattoliBM15}.

\paragraph{Names and Addresses.} It is natural to wonder whether one can refine the \EXAM by using addresses $\addr$ as a more precise form of names for jobs. It is possible, it is enough to modify the \EXAM as to extend at each step the name/address. For instance, transition $\toexamap$ would become:
\begin{center}
  $\begin{array}{|c@{\SEP}c@{\SEP}c||@{\hspace{6pt}}c@{\hspace{6pt}}||c@{\SEP}c@{\SEP}c|}
   \hline
  \textsc{Ap.} & \textsc{Pool} & \textsc{Env}&&  \textsc{Ap.} & \textsc{Pool} & \textsc{Env}
  \\
    \hline
    \hline
        \boapr
      & \cplace{\addr}{\tm\,\tmtwo}{\stack}\pop\pool
      & \genv
    & \toexamap &
        \boapr
      & \cplace{l\cons\addr}{\tm}{\tmtwo\cons\stack}\push\pool
      & \genv
\\
      \hline
  \end{array}
  $
\end{center}
Then a $\beta$-transition of address $\addr$ in a reachable state $\state$ corresponds exactly to a $\beta$-redex of address $\addr$ in $\decode\state$. We refrained from adopting addresses as names, however, because this is only useful for proving the $\beta$-reflection property of the \EXAM, the machine does not need such an additional structure for its functioning.
\section{Further Pool Templates}
Here we quickly discuss how by changing the pool template of the \EXAM we can implement different forms of evaluation, justifying the design choice of presenting pools as an interface.

\paragraph{Least Level.} Another sub-strategy of external reduction that computes $\beta$-normal forms is provided by \emph{least level reduction} $\toll$, a notion from the linear logic literature. Picking a redex of minimal level, where the level is the number of arguments in which the redex is contained, is another predicate (similarly to the leftmost one) that ensures externality of an outermost redex. Note that the $\Omega$ redex in \refeq{outermost-not-normalizing} (page \pageref{eq:outermost-not-normalizing}) is not of minimal level (it has level $1$ while the redex involving $\varthree$ has level $0$). Least level reduction is non-deterministic but diamond. For instance the two redexes (of level 1) in $\var (\Id \vartwo) (\Id\varthree)$ are both least level. Note that the leftmost redex might not be least level, as in $\var (\var (\Id \vartwo)) (\Id\varthree)$, where the leftmost redex is $\Id \vartwo$ and has level 2, while $\Id\varthree$ has level 1.

By slightly changing the pool template, we can turn the \LEXAM into a machine for least level evaluation. The key observation is that when new jobs are created, which is done only by transition $\toexamvar$, they all have level $n+1$ where $n$ is the level of the active job. To obtain a machine that processes job by increasing levels, then, it is enough to add the new jobs \emph{at the end of the pool}, rather than at the beginning. This template is an example where dropping (which pushes an element on top of the list of jobs) and adding (which adds at the end) are \emph{not} the same operation.

\paragraph{Fair Template.} Another interesting template is the one where pools are lists and dropping always adds at the end of the list. In this way the \EXAM is \emph{fair}, in the sense that even when it diverges, it keeps evaluating all jobs. This kind of strategies are of interest for infinitary $\l$-calculi, where one wants to compute all branches of an infinite normal form, instead of being stuck on one.

\section{Conclusions}
This paper studies two simple ideas and applies them to the paradigmatic case of strong call-by-name evaluation. Firstly, avoiding \emph{backtracking on the search for redexes} by introducing \emph{jobs} for each argument and jumping to the next job when one is finished. Secondly, modularizing the scheduling of jobs via a \emph{pool interface} that can be instantiated by various concrete schedulers, called \emph{pool templates}.

The outcome of the study is a compact, modular, and---we believe---elegant abstract machine for strong evaluation. In particular, we obtain the simplest machine for leftmost evaluation in the literature. Our study also gives a computational interpretation to the diamond non-determinism of strong call-by-name.

\correction{For the sake of simplicity, our study extends the MAM, which implements weak head reduction using global environments. Our technique, however, is reasonably modular in the underlying machine/notion of environment. One can, indeed, replace the MAM with Krivine abstract machine (KAM), which instead uses local environments, by changing only the fact that the jobs of the \EXAM have to carry their own local environment. Similarly, the technique seems to be adaptable to the CEK or other machines for weak evaluation. It would be interesting to compare the outcome of these adaptations with existing machines for strong call-by-value \cite{DBLP:conf/rta/BiernackaC19,DBLP:conf/lics/AccattoliCC21,DBLP:conf/ppdp/BiernackaCD21} or strong call-by-need \cite{DBLP:conf/aplas/BiernackaBCD20,DBLP:journals/pacmpl/BiernackaCD22}.}


\bibliographystyle{splncs04}
\bibliography{main.bbl}

\newpage
\appendix
\section{Proof Appendix}
This proof appendix has a section for each section of the paper, containing the omitted proofs.

\section{Normal Forms and the Importance of Being External}

\gettoappendix{prop:ext-diamond}
\begin{proof}\hfill
\begin{enumerate}
\item We show that every leftmost context $\lctx$ is an external context and every neutral context $\nctx$ is a rigid context. Note that every neutral term $\neu$ is a rigid term. By mutual induction on $\lctx$ and $\nctx$. Cases of $\lctx$:
\begin{itemize}
\item \emph{Empty}, that is, $\lctx=\ctxhole$. Immediate.
\item \emph{Abstraction}, that is, $\lctx=\la\var\lctxtwo$. By \ih, $\lctxtwo$ is external, and so is $\lctx$.
\item \emph{Left of an application}, that is, $\lctx=\nctx\tm$. By \ih, $\nctx$ is rigid, and so $\lctx$ is external.
\item \emph{Right of an application}, that is, $\lctx=\neu\lctxtwo$. By \ih, $\lctxtwo$ is external, and so $\lctx$ is external (because $\neu$ is a rigid term).
\end{itemize}
Cases of $\nctx$:
\begin{itemize}
\item \emph{Left of an application}, that is, $\nctx=\nctxtwo\tm$. By \ih, $\nctxtwo$ is rigid, and so $\nctx$ is rigid.
\item \emph{Right of an application}, that is, $\nctx=\neu\lctx$. By \ih, $\lctx$ is external, and so $\nctx$ is rigid (because $\neu$ is a rigid term).
\end{itemize}

%

\item
It suffices to show that if the source term $\tm$ can be written
as $\tm = \ctx\ctxholep{(\lam{\var}{\tmtwo})\,\tmthree}$
and also as $\tm = \ctx'\ctxholep{(\lam{\vartwo}{\tmtwo'})\,\tmthree'}$
where $\ctx$ and $\ctx'$ are both applicative (resp. external, rigid) contexts
then either they are equal ($\ctx = \ctx'$) or they are disjoint.
Proceed by induction on the derivation that $\ctx$ is an applicative
(resp. external, rigid) context.
\begin{itemize}
\item
  \emph{Applicative, base case, $\ctx = \ctxhole$}.
  Then $(\lam{\var}{\tmtwo})\,\tmthree = \ctx'\ctxholep{(\lam{\vartwo}{\tmtwo'})\,\tmthree'}$
  where $\ctx'$ is applicative,
  so necessarily $\ctx' = \ctxhole$ and we have $\ctx = \ctx'$ as required.
\item
  \emph{Applicative, application case, $\ctx = \ctx_1\,\tmfour$}.
  Symmetrically to the previous case, $\ctx'$ cannot be empty,
  so it is an application, \ie $\ctx' = \ctx'_1\,\tmfour$,
  and we conclude by \ih
\item
  \emph{External, applicative case, $\ctx = \apctx$}.
  Then there is a $\beta$-redex at the head of $\tm$.
  Note that $\ctx'$ cannot be a spinal external context
  (as this would imply that the head of $\tm$ is a variable, not a $\beta$-redex)
  nor an abstraction
  (as this would imply that the head of $\tm$ is an abstraction, not a $\beta$-redex).
  So $\ctx'$ must also be an applicative context, $\ctx' = \apctx'$, and we conclude by \ih
\item
  \emph{External, rigid case, $\ctx = \soctx$}.
  Then there is a variable at the head of $\tm$.
  Note that $\ctx'$ cannot be an applicative external context
  (as this would imply that the head of $\tm$ is a $\beta$-redex, not a variable)
  nor an abstraction
  (as this would imply that the head of $\tm$ is an abstraction, not a variable).
  So $\ctx'$ must also be a spinal external context, $\ctx' = \soctx'$, and we conclude by \ih
\item
  \emph{External, abstraction case, $\ctx = \lam{\var}{\boctx_1}$}.
  Then, reasoning as in the two previous cases,
  $\ctx'$ must also be an abstraction $\ctx' = \lam{\var}{\boctx'_1}$,
  and we conclude by \ih
\item
  \emph{Rigid, right case, $\ctx = \aptm\,\boctx$}.
  If $\ctx'$ is also built with the right case, \ie $\ctx' = \aptm\,\boctx'$
  then we conclude by \ih.
  If $\ctx'$ is built with the left case, \ie $\ctx' = \soctx\,\tm$
  then we have that $\ctx$ and $\ctx'$ are disjoint contexts.
\item
  \emph{Rigid, left case, $\ctx = \soctx\,\tm$}.
  Symmetric to the previous case.\qed
\end{itemize}
\end{enumerate}
\end{proof}

\section{Preliminaries: Abstract Machines}

\begin{lemma}[One-step simulation]
  \label{l:one-step-simulation}
  Let $\mach=(\States, \tomach, \compilrel\cdot\cdot, \decode\cdot)$ be a machine and $\tostrat$ be a strategy forming an implementation system. 
  For any state $\state$ of $\mach$, if $\decode\state \tostrat \tmtwo$ then there is a state $\statetwo$ of $\mach$ such that $\state \tomacho^*\tomachb \statetwo$ and $\decode{\statetwo} = \tmtwo$.
\end{lemma}

\begin{proof}
  For any state $\state$ of $\mach$, let $\nfo{\state}$ be a normal form of $\state$ with respect to $\tomacho$: such a state exists because overhead transitions terminate (\refpoint{def-overhead-terminate}).
  Since $\tomacho$ is mapped on identities (\refpoint{def-overhead-transparency}), one has $\decode{\nfo{\state}} = \decode\state$. 
The hypothesis then is $\decode\state = \decode{\nfo{\state}} \tostrat \tmtwo$, and by $\beta$-reflection (\refpoint{def-beta-reflection}) there is $\statetwo$ such that $\nfo{\state} \tomachb \statetwo$ and $\decode\statetwo = \tmtwo$.\qed
\end{proof}

\gettoappendix{thm:abs-impl}
\begin{proof}
  According to \refdef{implem}, given a $\lambda$-term $\tm$, we have to show that:
  \begin{enumerate}
    \item \label{p:exec-to-deriv} \emph{Runs to evaluations with $\beta$-matching}: for any $\mach$-run $\exec: \compilrel\tm\state \tomachine^* \statetwo$ there exists a $\tostrat$-evaluation $\deriv: \tm \tostrat^* \decode\statetwo$ such that $\size\deriv = \sizebeta\exec$.

    \item \label{p:deriv-to-exec} \emph{Evaluations to runs with $\beta$-matching}: for every $\tostrat$-evaluation $\deriv: \tm \tostrat^* \tmtwo$ there exists a $\mach$-run $\exec: \compilrel\tm\state \tomachine^* \statetwo$ such that $\decode\statetwo = \tmtwo$ and $\size\deriv = \sizebeta\exec$.
  \end{enumerate}

  \paragraph{Proof of \refpoint{exec-to-deriv}}  By induction on $\sizebeta\exec \in \nat$.
  
  If $\sizebeta\exec = 0$ then $\exec \colon \compilrel\tm\state \tomacho^* \statetwo$ and hence $\decode\state = \decode\statetwo$ by overhead transparency (\refpoint{def-overhead-transparency} of \refdef{implementation}).
  Moreover, $\tm = \decode{\state} = \decode{\statetwo}$ by the initialization constraint, therefore we are done by taking the empty (\ie without steps) evaluation $\deriv$ with starting (and end) term $\tm$.
  
  Suppose $\sizebeta\exec > 0$: then, $\exec \colon \compilrel\tm\state \tomachine^* \statetwo$ is the concatenation of a run $\execp \colon \compilrel{\tm}\state \tomachine^* \statethree$ followed by a run $\execpp \colon \statethree \tomachb \statefour \tomacho^* \statetwo$.
  By \ih applied to $\execp$, there exists an evaluation $\derivp \colon \tm \tostrat^* \decode\statethree$ with $\sizebeta\execp = \size\derivp$.
  By $\beta$-projection (\refpoint{def-beta-projection} of \refdef{implementation}) and overhead transparency (\refpoint{def-overhead-transparency} of \refdef{implementation}) applied to $\execpp$, one has $\derivpp \colon \decode\statethree \tostrat \decode\statefour = \decode\statetwo$.
  Therefore, the evaluation 
  $\deriv$ defined as the concatenation of $\derivp$ and $\derivpp$ is such that $\deriv \colon \tm \tostrat^* \decode\statetwo$ and $\size\deriv = \size\derivp + \size\derivpp = \sizebeta\execp + 1 = \sizebeta\exec$.
 
  \paragraph{Proof of \refpoint{deriv-to-exec}}  By induction on $\size\deriv \in \nat$.

  If $\size\deriv = 0$ then $\tm = \tmtwo$.
 By the initialization constraint, one has $\decode{\state} = \tm$.
  We are done by taking the empty (\ie without transitions) run $\exec$ with initial (and final) state $\state$.
  
  Suppose $\size\deriv > 0$: so, $\deriv\colon \tm \tostrat^* \tmtwo$ is the concatenation of an evaluation $\derivp \colon \tm \tostrat^* \tmtwop$ followed by the step $\tmtwop \tostrat \tmtwo$.
  By \ih, there exists a $\mach$-run $\execp\colon \compilrel\tm\state \tomachine^* \statethree$ such that $\decode\statethree = \tmtwop$ and $\size\derivp = \sizebeta\execp$.
By  one-step simulation (\reflemma{one-step-simulation}, since $\decode{\statethree} \tostrat \tmtwo$ and $\tostrat$ and $\mach$ form an implementation system), there is a state $\statetwo$ of $\mach$ such that $\statethree \tomacho^*\tomachb \statetwo$ and $\decode\statetwo = \tmtwo$.
  Therefore, the run $\exec \colon \compilrel\tm\state \tomachine^*\statethree \tomacho^*\tomachb \statetwo$ is such that $\sizebeta\exec = \sizebeta\execp +1 = \size\derivp + 1 = \size\deriv$. \qed
\end{proof}

\section{The External Abstract Machine}

\gettoappendix{l:inn-extension-approx}
\begin{proof}
We show the following statement, strengthened with
a similar property for rigid approximants:
\begin{enumerate}
\item
  \emph{If $\soapr$ is a rigid approximant
  then $\soapr\ctxholep{\lam{\var}{\boctxhole\name}}_\name$
  and $\soapr\ctxholep{\var\,\boctxhole{\nametwo_1}\mydots\boctxhole{\nametwo_n}}_\name$ are
  also rigid approximants.}
  
  \item \emph{If $\boapr$ is an  approximant
then $\boapr\boctxholep\name{\lam{\var}{\boctxhole\name}}$
and $\boapr\boctxholep\name{\var\,\boctxhole{\nametwo_1}\mydots\boctxhole{\nametwo_n}}$ are
also approximants.}
\end{enumerate}
The proof is by mutual induction on  $\boapr$ and $\soapr$.
The case $\soapr = \var$ is immediate as there are no names in $\var$.
The remaining cases are straightforward by \ih The only interesting case is
when $\boapr$ is a named hole, \ie $\boapr = \boctxhole\namethree$.
If $\namethree = \name$, it suffices to remark that $\lam{\var}{\ctxhole_\name}$
and $\var\,\boctxhole{\nametwo_1}\mydots\boctxhole{\nametwo_n}$ are both approximants.
If $\namethree \neq \name$, it is immediate since
$\boctxhole\namethree\boctxholep\name{\lam{\var}{\ctxhole_\name}} = \boctxhole\namethree$ and $
\boctxhole\namethree\boctxholep\name{\var\,\boctxhole{\nametwo_1}\mydots\boctxhole{\nametwo_n}} =\boctxhole\namethree$.\qed
\end{proof}

\gettoappendix{l:approx-no-names-normal}
\begin{proof}
First of all, we strengthen the statement adding a part about rigid approximants:
\begin{enumerate}
\item A rigid approximant $\soapr$ without names is a neutral term;
\item An approximant $\boapr$ without names is a normal form.
\end{enumerate}
Then the proof is a straightforward mutual induction on $\soapr$ and $\boapr$.\qed
\end{proof}

The next lemma collects the basic properties of the read-back that we use in the proofs.
\begin{lemma}[Properties 2]
\label{l:new-properties_decoding} 
\begin{enumerate}
\item \label{p:new-properties_decoding-one}
Let $\suppset$ be a uniquely named set of jobs. Then 
$\decpool{(\soapr\,\boapr)}\suppset = \decpool\soapr\suppset\,\decpool\boapr\suppset$.

\item \label{p:new-properties_decoding-two}
Let $\suppset$ and $\suppsettwo$ be uniquely named sets of jobs such that $\names{\suppset}\cap\names{\suppsettwo}=\emptyset$ then  $\decpool{\ntm}{\suppset\cup\suppsettwo} = \decpool{(\decpool{\ntm}{\suppset})}{\suppsettwo}$.

\item \label{p:new-properties_decoding-three}
$\decst{(\la\var\tm)\tmtwo}\stack \towh \decst{\tm\isub\var\tmtwo}\stack$.

\item \label{p:new-properties_decoding-four}
$\decenv\tm\env \isub\var{\decenv\tmtwo\env} =\decenv{\tm\isub\var\tmtwo}\env$.

\item \label{p:new-properties_decoding-five}
If $\var\notin\fv{\place_\name}$ then $\decenv{\place_\name}{\esub\var\tm\cons\env} = \decenv{\place_\name}{\env}$.
\item \label{p:new-properties_decoding-six}
Let $\suppset$ be a uniquely named set of jobs. If $\var\notin\fv{\suppset}$ then $\decenv{\suppset}{\esub\var\tm\cons\env} = \decenv{\suppset}{\env}$.
\item \label{p:new-properties_decoding-seven}
If $\env$ verifies the local scope invariant and $\var\in\dom\env$ then $\decenv\var\env=\decenv{\env(\var)}\env$.


\end{enumerate}
\end{lemma}

\begin{proof}
Straightforward.
\end{proof}

\gettoappendix{prop:exam-diamond}
\begin{proof}
Note that manipulations of the approximant of different names commute, since they are defined via plugging on different names. Thus any two transitions from a state $\state$ can be done in any order without changing the result except for two $\beta$-transitions, which are the only transitions that extend the environment. Then one has that they commute up to $\approx$, which has been defined exactly for this purpose.\qed
\end{proof}
\section{Runs to Evaluations}

\gettoappendix{prop:overhead-transparency}
\begin{proof}
In the proof, we repeatedly use the following fact referring to it with '$*$':
  \[
    \begin{array}{llllllllllll}    
      \decode{ \examstate\boapr {\job_\name \pop\pool} \genv   } 
    & = &
      \decpool{\boapr}{\decenv{ \set{ \job_\name } \cup \supp\pool }\env }
    \\
    & =_{\reflemmaeqp{new-properties_decoding}{two}} &
      \decpool{\decpool\boapr{\decenv{ \set{ \job_\name }}\env}} {\decenv\pool\env }
    \\
    & = &
      \decpool{\boapr\tsub{\name}{\decenv{ \job_\name }\env}} {\decenv\pool\env }
    
    & = &
      \decpool\boapr{\decenv\pool\env }\tsub{\name}{\decenv{ \job_\name }\env}
    \end{array}
  \]
And similarly one obtains that  $\decode{ \examstate\boapr {\job_\name \push\pool} \genv} = \decpool\boapr{\decenv\pool\env }\tsub{\name}{\decenv{ \job_\name }\env}$, also to be referred to with '$*$'.

The proof is by cases on the kind of transition.
\begin{itemize}
\item \emph{Application search ($\toexamap$).}
  \[
    \begin{array}{rcl}
      \decode{\examstate{
        \boapr 
      }{
        \cplace{\name}{\tm\,\tmtwo}{\stack}\pop\pool
      }{
        \genv
      }}
    & =_* &
      \decpool\boapr{\decenv\pool\env }\tsub{\name}{\decenv{ \decst{(\tm\,\tmtwo)}\stack }\env}
    \\
    & = &
      \decpool\boapr{\decenv\pool\env }\tsub{\name}{\decenv{ \decst\tm{\tmtwo\cons\stack} }\env}
    \\
    & =_* &
      \decode{\examstate{
        \boapr
      }{
        \cplace{\name}{\tm}{\tmtwo\cons\stack}\push\pool
      }{
        \genv
      }}
    \end{array}
  \]
  
\item \emph{Abstraction search ($\toexamlam$).}
  Then:
  \[
    \begin{array}{rcll}
      \decode{
        \examstate{
          \boapr
        }{
          \cplace{\name}{\lam{\var}{\tm}}{\stempty}\pop\pool
        }{
          \genv
        }
      }
      & =_* &
 \decpool\boapr{\decenv\pool\env }\tsub{\name}{ \lam{\var}{\decenv\tm\env} }
    \\
      & = &
 \left(\decpool\boapr{\decenv\pool\env }\tsub{\name}{ \lam{\var}{\boctxhole\name} } \right)\tsub{\name}{\decenv\tm\env}
    \\
      & = &
 \decpool{\boapr\tsub{\name}{ \lam{\var}{\boctxhole\name} }} {\decenv\pool\env } \tsub{\name}{\decenv\tm\env}
    \\
      & =_* &
      \decode{
        \examstate{
          \boapr\tsub{\name}{\lam{\var}{\boctxhole\name}}
        }{
          \cplace{\name}{\tm}{\stempty}\push\pool
        }{
          \genv
        }
      }
    \end{array}
  \]
  
\item \emph{Substitution ($\toexamsub$).} Let $\envtwo \defeq \env_1\cons\esub\var\tm\cons\env_2$. Then:
  \[
    \begin{array}{rcll}
      &&
      \decode{
        \examstate{
          \boapr
        }{
          \cplace{\name}{\var}{\stack}\pop\pool
        }{
          \envtwo
        }
      }
    \\
      & =_* &
 \decpool\boapr{\decenv\pool\env }\tsub{\name}{ \decenv{\decst\var\stack}\envtwo }
    \\
      & = &
 \decpool\boapr{\decenv\pool\env }\tsub{\name}{ \decst{\decenv\var\envtwo}{\decenv\stack\envtwo} }
    \\
      & =_{\reflemmaeqp{new-properties_decoding}{seven}} &
 \decpool\boapr{\decenv\pool\env }\tsub{\name}{ \decst{\decenv{\tm}\envtwo}{\decenv\stack\envtwo} }
    \\
      & =_\alpha &
 \decpool\boapr{\decenv\pool\env }\tsub{\name}{ \decst{\decenv{\rename\tm}\envtwo}{\decenv\stack\envtwo} }
    \\
      & = &
 \decpool\boapr{\decenv\pool\env }\tsub{\name}{ \decenv{\decst{\rename\tm}\stack}\envtwo }
    \\
      & =_* &
      \decode{
        \examstate{
          \boapr
        }{
          \cplace{\name}{\rename{\tm}}{\stack}\push\pool
        }{
          \envtwo
        }
      }
    \end{array}
  \]
  
\item \emph{Variable search ($\toexamvar$).}
  Let $\nametwo_1,\mydots,\nametwo_n$ be fresh names. Then:
  \[
    \begin{array}{rcll}
    &&
      \decode{
        \examstate{
          \boapr
        }{
          \cplace{\name}{\var}{\tm_1\cons\mydots\cons\tm_n}\pop\pool
        }{
          \genv
        }
      }
    \\
    & = &
      \decenv{
        (\decpool{
          \boapr\tsub{\name}{\var\,\tm_1\mydots\tm_n}
        }{
          \pool
        })
      }{
        \genv
      }
    \\
          & =_* &
 \decpool\boapr{\decenv\pool\env }\tsub{\name}{ \var\,\decenv{\tm_1}\env\mydots\decenv{\tm_n}\env} 
    \\
      & = &
 \left(\decpool\boapr{\decenv\pool\env }\tsub{\name}{ \var\,\boctxhole{\nametwo_1}\mydots\boctxhole{\nametwo_n} } \right)\tsub{\nametwo_1}{\decenv{\tm_1}\env}\mydots\tsub{\nametwo_n}{\decenv{\tm_n}\env}
    \\
      & = &
 \decpool{\boapr\tsub{\name}{ \var\,\boctxhole{\nametwo_1}\mydots\boctxhole{\nametwo_n} }} {\decenv\pool\env } 
 \tsub{\nametwo_1}{\decenv{\tm_1}\env}\mydots\tsub{\nametwo_n}{\decenv{\tm_n}\env}
    \\
      & = &
 \decpool{\boapr\tsub{\name}{ \var\,\boctxhole{\nametwo_1}\mydots\boctxhole{\nametwo_n} }
 \tsub{\nametwo_1}{\decenv{\tm_1}\env}\mydots\tsub{\nametwo_n}{\decenv{\tm_n}\env}
 } {\decenv\pool\env } 
    \\
      & = &
 \decpool{\decpool{
 \boapr\tsub{\name}{ \var\,\boctxhole{\nametwo_1}\mydots\boctxhole{\nametwo_n} }
 } {\set{\cplace{\nametwo_1}{\tm_1}{\stempty},\mydots,\cplace{\nametwo_n}{\tm_n}{\stempty}}}
 } {\decenv\pool\env } 
    \\
      & =_{\reflemmaeqp{new-properties_decoding}{two}} &
 \decpool{
 \boapr\tsub{\name}{ \var\,\boctxhole{\nametwo_1}\mydots\boctxhole{\nametwo_n} }
 }
  {\decenv{\set{\cplace{\nametwo_1}{\tm_1}{\stempty},\mydots,\cplace{\nametwo_n}{\tm_n}{\stempty}} \cup\supp\pool}\env } 
    \\
    & = &
      \decode{
        \examstate{
          \boapr\tsub{\name}{\var\,\boctxhole{{\nametwo_1}\mydots\boctxhole{\nametwo_n}}}
        }{
         \cplace{\nametwo_1}{\tm_1}{\stempty}\cons\mydots\cons\cplace{\nametwo_n}{\tm_n}{\stempty}\add\pool
        }{
          \genv
        }
      }
    \end{array}
  \]
  \end{itemize}
  \qed
\end{proof}

\gettoappendix{l:exam_invariant}
\begin{proof}
\correction{By induction on the length $k$ of the execution leading to the reachable state $\state$. If $k=0$ then the state is initial, and the freshness and freeness invariants hold because the term in initial states is well-named. For k>0,} it suffices to show that each transition preserves the invariants. For each transition, we only mention the non-immediate invariants. The following case analysis shows that the proof of freshness and bijection uses uniqueness.
\begin{enumerate}
\item \emph{Application search ($\toexamap$).}
All invariants trivially hold.

\item \emph{Beta ($\toexambeta$).}
  For freshness,
  note that $\var$ has a binding occurrence $\esub{\var}{-}$ on the RHS
  and a binding occurrence $\lam{\var}{-}$ on the LHS,
  so by Freshness there are no other binding occurrences of $\var$.

  For local~scope, note that by local~scope of the LHS
  there are no occurrences of $\var$ other than in $\tm$,
  so there are also no occurrences of $\var$ in $\genv$.
  
\item \emph{Substitution ($\toexamsub$).}
  For freshness, note that $\rename{\tm}$ is a fresh renaming of $\tm$,
  so the binding occurrences of variables in $\rename{\tm}$ are different from each other
  and from any other binding occurrences in the state.
 
  For local~scope, observe that $\rename{\tm}$ is a fresh renaming of $\tm$,
  so if $\lam{\vartwo}{\tmthree}$ is a sub-term of $\tm$ there are no occurrences
  of $\vartwo$ other than in $\tmthree$.

\item \emph{Abstraction search ($\toexamlam$).}
  For freshness, 
  note that $\name$ appears exactly once in $\boapr$ by uniqueness for the LHS,
  so $\var$ has one binding occurrence in $\boapr\tsub{\name}{\lam{\var}{\boctxhole\name}}$
  on the RHS. Moreover, since $\var$ has one binding occurrence on the pool of the
  LHS, by freshness there are no other occurrences of $\var$.

  For uniqueness note, as before, that by uniqueness of the LHS,
  we have that $\name$ appears exactly once in $\boapr$, and so it appears exactly once in $\boapr\tsub{\name}{\lam{\var}{\boctxhole\name}}$. Since $\name \in \names{\cplace{\name}{\tm}{\stempty}\drop\pool}$ the bijection is preserved.

\item \emph{Variable search ($\toexamvar$).}
  For uniqueness, note that there is exactly one occurrence of $\name$ in $\boapr$ by uniqueness of the LHS.
  So, given that $\nametwo_1,\mydots,\nametwo_n$ are fresh,
  for each $1 \leq i \leq n$
  there is exactly one occurrence of $\nametwo_i$ in
  $\boapr\tsub{\name}{\var\,\boctxhole{\nametwo_1}\mydots\boctxhole{\nametwo_n}}$. Since 
  \[\nametwo_i\in\names{\cplace{\nametwo_1}{\tm_1}{\stempty}\cons\mydots\cons\cplace{\nametwo_n}{\tm_n}{\stempty}\pro\pool}\] the bijection is also preserved. For freeness, it follows from freeness for the LHS and the fact that $\var\notin\dom\genv$, which is one of the conditions for the transition to apply.\qed
\end{enumerate}
\end{proof}

\gettoappendix{l:pure_decoding}
\begin{proof}
\hfill
\begin{enumerate}
\item The statement has to be strengthened adding a part about rigid approximants, as follows:
\begin{enumerate}
\item \emph{Let $\soapr$ be an rigid appoximant and let $\pool$ be a pool such that $\fn{\soapr} \subseteq \names\pool$. Then $\decpool{\soapr}{\pool}$ is a rigid term without any free occurrences of names.}

\item \emph{Let $\boapr$  be an  approximant  and let $\pool$ be a pool such that $\fn{\boapr} \subseteq \names\pool$. Then $\decpool{\boapr}{\pool}$ is a term  without any free occurrences of names.}
\end{enumerate}
Then the proof is a straightforward mutual induction on $\boapr$ and $\soapr$.

\item It follows from the first point and from the fact that from the bijection invariant (\reflemma{exam_invariant}.3) any reachable state is in the hypothesis of the first point.
\qed
\end{enumerate}
\end{proof}

\gettoappendix{l:almost_pure_decoding}
\begin{proof}
As usual, we first strengthen the statement adding a part about rigid approximants. The new statement is:

\emph{Let $\suppset$ be a set of uniquely named jobs.}
\begin{enumerate}
\item \emph{Let $\soapr$ be a rigid approximant with no repeated names
such that $\fn{\soapr} \setminus \names\suppset = \set{\name}$.
Then $\decpool{\soapr}{\suppset}\tsub{\name}{\ctxhole}$ is a rigid context.}

\item \emph{Let $\boapr$ be an approximant with no repeated names
such that $\fn{\boapr} \setminus \names\suppset = \set{\name}$.
Then $\decpool{\boapr}{\suppset}\tsub{\name}{\ctxhole}$  is an external context.}
\end{enumerate}

The proof is by mutual induction on $\soapr$ and $\boapr$.     
    Cases of $\soapr$:
  \begin{enumerate}
\item \emph{Variable case, $\soapr = \var$.}
  This case is impossible, as by hypothesis we know that $\name \in \fn{\var} = \emptyset$.
  
\item \emph{Application case, $\soapr = \soaprtwo\,\boapr$.}
Two cases:
  \begin{enumerate}
  \item
    If $\name \in \fn{\soaprtwo}$, then $\fn{\soaprtwo} \setminus \names\suppset = \set{\name}$ and we can apply the \ih
    Then:
    \[
      \begin{array}{rlll}
            \decpool{(\soaprtwo\,\boapr)}{\suppset}\tsub{\name}{\ctxhole}
      & =_{\reflemmaeqp{new-properties_decoding}{one}} & (\decpool\soaprtwo\suppset\,\decpool\boapr\suppset)\tsub{\name}{\ctxhole}
      
      \\
      & = & \decpool{\soaprtwo}{\suppset}\tsub{\name}{\ctxhole}\,\decpool{\boapr}{\suppset}
      \end{array}
    \]
    By \ih, $\decpool{\soaprtwo}{\suppset}\tsub{\name}{\ctxhole}$ is a rigid~context. Since $\soapr$ has no repeated names, $\name\notin\fn\boapr$, so that $\fn\boapr \subseteq \names\suppset$ and we can apply \reflemma{pure_decoding}.1, obtaining that $\decpool{\boapr}{\suppset}$ is a term (without free occurrences of names).
    So indeed $\decpool{\soaprtwo}{\suppset}\tsub{\name}{\ctxhole}\,\decpool{\boapr}{\suppset}$ is a rigid~context.
  \item
    If $\name \in \fn{\boapr}$, then $\fn{\boapr} \setminus \names\suppset = \set{\name}$ and we can apply the \ih
    Then:
    \[
      \begin{array}{rlll}
            \decpool{(\soaprtwo\,\boapr)}{\suppset}\tsub{\name}{\ctxhole}
      & =_{\reflemmaeqp{new-properties_decoding}{one}} & (\decpool\soaprtwo\suppset\,\decpool\boapr\suppset)\tsub{\name}{\ctxhole}

      \\
      & = & \decpool{\soaprtwo}{\suppset}\,\decpool{\boapr}{\suppset}\tsub{\name}{\ctxhole}

      \end{array}
    \]
By \ih,
    $(\decpool{\boapr}{\suppset})\tsub{\name}{\ctxhole}$ is an external context.  Since $\boapr$ has no repeated names, $\name\notin\fn\soaprtwo$, so that $\fn\soaprtwo \subseteq \names\suppset$ and we can apply \reflemma{pure_decoding}.1, obtaining that $\decpool{\soaprtwo}{\suppset}$ is a term (without free occurrences of names).
    So indeed $\decpool{\soaprtwo}{\suppset}\,(\decpool{\boapr}{\suppset})\tsub{\name}{\ctxhole}$ is an external context.
  \end{enumerate}
\end{enumerate}

Cases of $\boapr$:
\begin{enumerate}
\item \emph{Name case, $\boapr = \boctxhole\nametwo$.}
  Then since $\fn{\boctxhole\nametwo} \setminus \names\suppset = \set{\name}$
  we have that $\nametwo = \name$ and $\suppset = \pempty$,
  and $\ctxhole$ is indeed an external context.
\item \emph{Rigid case, $\boapr = \soapr$.}
  Immediate by \ih
\item \emph{Abstraction case, $\boapr = \lam{\var}{\boaprtwo}$.}
  By \ih $\decpool{\boaprtwo}{\suppset}\tsub{\name}{\ctxhole}$ is an external context,
  so $\decpool{(\lam{\var}{\boaprtwo})}{\suppset}\tsub{\name}{\ctxhole} =
      \lam{\var}{(\decpool{\boaprtwo}{\suppset}\tsub{\name}{\ctxhole})}$
  is also an external context.\qed
    \end{enumerate}
\end{proof}

\gettoappendix{thm:exam-beta-projection}
\begin{proof}
In the proof, we repeatedly use the following fact referring to it with '$*$':
  \[
    \begin{array}{llllllllllll}    
      \decode{ \examstate\boapr {\job_\name \pop\pool} \genv   } 
    & = &
      \decpool{\boapr}{\decenv{ \set{ \job_\name } \cup \supp\pool }\env }
    \\
    & =_{\reflemmaeqp{new-properties_decoding}{two}} &
      \decpool{\decpool\boapr{\decenv{ \set{ \job_\name }}\env}} {\decenv\pool\env }
    \\
    & = &
      \decpool{\boapr\tsub{\name}{\decenv{ \job_\name }\env}} {\decenv\pool\env }
    
    & = &
      \decpool\boapr{\decenv\pool\env }\tsub{\name}{\decenv{ \job_\name }\env}
    \end{array}
  \]
And similarly one obtains that  $\decode{ \examstate\boapr {\job_\name \push\pool} \genv} = \decpool\boapr{\decenv\pool\env }\tsub{\name}{\decenv{ \job_\name }\env}$, also to be referred to with '$*$'.

The proof of $\beta$-projection is given by:
  \[
    \begin{array}{rcll}
    &&
      \decode{\examstate{
        \boapr
      }{
        \cplace{\name}{\lam{\var}{\tm}}{\tmtwo\cons\stack}\pop\pool
      }{
        \genv
      }}
    \\
    & =_* &
      \decpool\boapr{\decenv\pool\env }\tsub{\name}{\decenv{ \cplace{\name}{\lam{\var}{\tm}}{\tmtwo\cons\stack}}\env}
    \\
    & = &
      \decpool\boapr{\decenv\pool\env }\tsub{\name}{ \decst{((\lam{\var}{\decenv\tm\env})\,\decenv\tmtwo\env)}{\decenv\stack\env} }
    \\
    & \toext &
      \decpool\boapr{\decenv\pool\env }\tsub{\name}{ \decst{\decenv\tm\env \isub\var{\decenv\tmtwo\env}}{\decenv\stack\env} }
      & \reflemmaeq{almost_pure_decoding}\mbox{ and }\reflemmaeqp{new-properties_decoding}{three}
    \\
    & =_{\reflemmaeqp{new-properties_decoding}{four}} &
      \decpool\boapr{\decenv\pool\env }\tsub{\name}{ \decst{\decenv{\tm \isub\var\tmtwo}\env}{\decenv\stack\env} }
    \\
    & = &
      \decpool\boapr{\decenv\pool\env }\tsub{\name}{\decenv{ \cplace{\name}{\tm\isub\var\tmtwo}{\stack}}\env}
    \\
    & =_{\reflemmaeqp{new-properties_decoding}{five}} &
      \decpool\boapr{\decenv\pool\env }\tsub{\name}{\decenv{ \cplace{\name}{\tm}{\stack}}{\esub\var\tmtwo\cons \env}}
    \\
    & =_{\reflemmaeqp{new-properties_decoding}{six}} &
      \decpool\boapr{\decenv\pool{\esub\var\tmtwo\cons \env} }\tsub{\name}{\decenv{ \cplace{\name}{\tm}{\stack}}{\esub\var\tmtwo\cons \env}}
    \\
    & =_* &
      \decode{\examstate{
        \boapr
      }{
        \cplace{\name}{\tm}{\stack}\push\pool
      }{
        \esub{\var}{\tmtwo}\cons\genv
      }}
    \end{array}
  \]
\qed
\end{proof}

\gettoappendix{coro:exam-correctness}
\begin{proof}
Simply apply the theorem about sufficient conditions for implementations (\refthm{abs-impl}) using overhead transparency (\refprop{overhead-transparency}) and $\beta$-projection (\refthm{exam-beta-projection}) of the \EXAM.\qed
\end{proof}

\gettoappendix{prop:halt}
\begin{proof}
Note that given a job $\cplace\name\tm\stack$ there is a transition for every case of $\tm$, and that a job is removed from the pool without adding other jobs when $\tm$ is a variable not in the environment and $\stack$ is empty. \correction{Moreover, by the choice constraint on the selection operation, if the pool is non-empty a job gets selected.} Therefore, final states have empty pools, that is, they have shape $\examstate\boapr\pool\genv$ with $\supp\pool=\emptyset$. Note that with such pools the read back gives the approximant, since $\decpool{\boapr}{\decenv\pempty\env}= \boapr$. By the bijection invariant (\reflemma{exam_invariant}.3), $\boapr$ has no names, that is, $\fn\boapr=\emptyset$, and so $\boapr$ is a normal form $\nf$ by \reflemma{approx-no-names-normal}. 

If $\state \streq \statetwo$ then $\statetwo = \examstate\boapr\pool\envtwo$ with $\envtwo \approx \env$. Since $\envtwo \approx \env$ implies $\dom\env=\dom\envtwo$, we also have $\decode\statetwo = \boapr = \nf$.
\qed
\end{proof}

\subsection{Leftmost Runs to Leftmost Evaluations}

\gettoappendix{prop:lexam-complex-invariant}
\begin{proof}
By induction on the length $k$ of the execution leading to $\state$. Cases:
\begin{itemize}
\item \emph{Initial state}, that is, $k=0$. Then $\state = \examstate{\boctxhole{\name_1}}{\new{\cplace{\name_1}\tm\stempty}}\eempty$. Then $\ctx_{\mid\lctx\mid}^\state = \lctx$ is a non-applying leftmost context, and it is the only context that needs to be considered. 

\item \emph{Non-initial state}, that is, $k>0$. Then we analyze the last transition $\statetwo\tolexam\state$. Cases:
\begin{itemize}
\item For $\toexamap$, $\toexambeta$, and $\toexamsub$, note that they do not change the approximant, the environment, nor the pool behind the active job, and for the active job they do not change the name, which is the only relevant component for the contexts $\ctx^{\state}_{\nf_1,\mydots,\nf_{i-1}\mid\lctx\mid \tm_{i+1},\mydots,\tm_n}$. So one has $\ctx^{\state}_{\nf_1,\mydots,\nf_{i-1}\mid\lctx\mid \tm_{i+1},\mydots,\tm_n} = \ctx^{\statetwo}_{\nf_1,\mydots,\nf_{i-1}\mid\lctx\mid \tm_{i+1},\mydots,\tm_n}$ for all $i\in\set{1,\mydots,n}$, and the statement follows from the \ih

\item \emph{Abstraction search ($\toexamlam$).}
Then:
    \begin{center}
     $\begin{array}{rlllllll}
	\statetwo
     & = &
	\examstate\boapr{ \cplace{\name_1}{\lam{\var}{\tm}}{\stempty}\pop\pool} \genv
    \\
    & \toexamlam &
\examstate{ \boapr\tsub{\name_1}{\lam{\var}{\boctxhole{\name_1}}} }{ \cplace{\name_1}{\tm}{\stempty}\push\pool } \genv    & = &
	\state.
    \end{array}$
    \end{center}
By \ih, $\ctx^{\statetwo}_{\nf_1,\mydots,\nf_{i-1}\mid\lctx\mid \tm_{i+1},\mydots,\tm_n} = \boapr\tsub{\name_1}{\nf_1}\mydots\tsub{\name_{i-1}}{\nf_{i-1}}\tsub{\name_i}{\lctx} \tsub{\name_{i+1}}{\tm_{i+1}}\mydots\tsub{\name_{n}}{\tm_{n}}$ is a non-applying leftmost context, we have to show that:
\begin{enumerate}
\item $\ctx_{\mid\lctx\mid \tm_{2},\mydots,\tm_n}^{\state} = \boapr\tsub{\name_1}{\la\var\lctx}\tsub{\name_{2}}{\tm_{2}}\mydots\tsub{\name_{n}}{\tm_{n}}$ is a non-applying leftmost context, and that
\item $\ctx_{\nf_1,\mydots, \nf_{i-1}\mid \lctx\mid \tm_{i+1},\mydots,\tm_n}^{\state} = \boapr\tsub{\name_1}{\la\var\ctxholep{\nf_1}}\mydots\tsub{\name_{i-1}}{\nf_{i-1}}\tsub{\name_i}{\lctx} \tsub{\name_{i+1}}{\tm_{i+1}}\mydots\tsub{\name_{n}}{\tm_{n}}$ is a non-applying leftmost context, for $i$ such that $1<i\leq n$.
\end{enumerate}
For point 1, note that $\la\var\lctx$ is a non-applying leftmost context because $\lctx$ is, even if $\lctx$ is empty, and that:
\begin{center}
    $\begin{array}{rcllllllll}
\ctx_{\mid\lctx\mid \tm_{2},\mydots,\tm_n}^{\state} & = & \boapr\tsub{\name_1}{\la\var\lctx}\tsub{\name_{2}}{\tm_{2}}\mydots\tsub{\name_{n}}{\tm_{n}}
        & = &
    \ctx_{\mid \la\var\lctx\mid \tm_{2},\mydots,\tm_n}^{\statetwo}
    \end{array}$
  \end{center}
thus the statement follows from the \ih 

For point 2, note that $\la\var\nf_1$ is normal (because $\nf_1$ is) and that $\ctx_{\nf_1,\mydots, \nf_{i-1}\mid\lctx\mid \tm_{i+1},\mydots,\tm_n}^{\state} = \ctx_{\la\var\nf_1,\mydots, \nf_{i-1}\mid\lctx\mid \tm_{i+1},\mydots,\tm_n}^{\statetwo}$, which is a non-applying leftmost context by the \ih


\item \emph{Variable search ($\toexamvar$).}
  Let $k\geq 0$ and $\name_1,\mydots,\name_k$ be fresh names. Then:
    \begin{center}
     $\begin{array}{rlllllll}
	\statetwo
     & = &
	\examstate\boapr{ \cplace{\name_0}{\var}{\tm_1\cons\mydots\cons\tm_k}\pop\pool } \genv
    \\
    & \toexamvartwo &
      \examstate{ \boapr\tsub{\name_0}{ \var\,\boctxhole{\name_1}\mydots\boctxhole{\name_n} } }{ \cplace{\name_1}{\tm_1}{\stempty}\cons\mydots\cons\cplace{\name_k}{\tm_k}{\stempty} \add\pool } \genv
    & = &
	\state.
    \end{array}$
    \end{center}
    Let $\pool = \place_{\name_{k+1}}\cons\mydots\cons \place_{\name_n}$. 
    By \ih, we have that:
\begin{itemize}
\item $\ctx_{\mid\lctx\mid \tm_{k+1},\mydots,\tm_n}^{\statetwo} = \boapr\tsub{\name_0}{\lctx} \tsub{\name_{k+1}}{\tm_{k+1}}\mydots\tsub{\name_{n}}{\tm_{n}}$ is a non-applying leftmost context, and that
\item $\ctx_{\nf_0,\nf_{k+1},\mydots, \nf_{i-1}\mid\lctx\mid \tm_{i+1},\mydots,\tm_n}^{\statetwo} =  \boapr\tsub{\name_0}{\nf_0}\tsub{\name_{k+1}}{\nf_{k+1}}\mydots\tsub{\name_{i-1}}{\nf_{i-1}}\tsub{\name_i}{\lctx} \tsub{\name_{i+1}}{\tm_{i+1}}\mydots\tsub{\name_{n}}{\tm_{n}}$ is a non-applying leftmost context, for $i$ such that $k+1\leq i\leq n$.
\end{itemize}
    We have to show that:
\begin{enumerate}
\item 
\[\begin{array}{lllll}
&&\ctx_{\nf_1,\mydots, \nf_{i-1}\mid\lctx\mid \tm_{i+1},\mydots,\tm_k, \tm_{k+1},\mydots,\tm_n}^{\state} 
\\
& = &
\boapr\tsub{\name_0}{\var\,\boctxholep{\name_1}{\nf_1}\mydots\boctxholep{\name_{i-1}}{\nf_{i-1}}\boctxholep{\name_i}\lctx \boctxholep{\name_{i+1}}{\tm_{i+1}} \mydots \boctxholep{\name_{k}}{\tm_{k}} } \tsub{\name_{k+1}}{\tm_{k+1}}\mydots\tsub{\name_{n}}{\tm_{n}}
\end{array}\]
 is a non-applying leftmost context, for $i$ such that $1\leq i\leq k$, and that

\item 
\[\begin{array}{lllll}
&&\ctx_{\nf_1,\mydots, \nf_k,\mydots, \nf_{i-1}\mid\lctx\mid \tm_{i+1},\mydots,\tm_n}^{\state}
\\
& = &
\boapr\tsub{\name_0}{\var\,\boctxholep{\name_1}{\nf_1}\mydots\boctxholep{\name_k}{\nf_k}} \tsub{\name_{k+1}}{\nf_{k+1}}\mydots\tsub{\name_{i-1}}{\nf_{i-1}}\tsub{\name_i}{\lctx} \tsub{\name_{i+1}}{\tm_{i+1}}\mydots\tsub{\name_{n}}{\tm_{n}}
\end{array}\]
 is a non-applying leftmost context, for $i$ such that $k+1\leq i\leq n$.
\end{enumerate}
For point 1, note that 
\[\lctxtwo \defeq \var\,\boctxholep{\name_1}{\nf_1}\mydots\boctxholep{\name_{i-1}}{\nf_{i-1}}\boctxholep{\name_i}\lctx \boctxholep{\name_{i+1}}{\tm_{i+1}} \mydots \boctxholep{\name_{k}}{\tm_{k}}\]
is a non-applying leftmost context because $\lctx$ is, even if $\lctx$ is empty, so that 
\[\ctx_{\nf_1,\mydots, \nf_{i-1}\mid\lctx\mid \tm_{i+1},\mydots,\tm_k, \tm_{k+1},\mydots,\tm_n}^{\state} = \ctx_{\lctxtwo\mid \tm_{k+1},\mydots,\tm_n}^{\statetwo}.\]
Thus, the statement follows from the \ih

For point 2, note that 
\[\nf \defeq \var\,\boctxholep{\name_1}{\nf_1}\mydots\boctxholep{\name_k}{\nf_k}\]
is a normal form (because $\nf_j$ is, for $1\leq j \leq k$), so that 
\[\ctx_{\nf_1,\mydots, \nf_k,\mydots, \nf_{i-1}\mid\lctx\mid \tm_{i+1},\mydots,\tm_n}^{\state} = \ctx_{\nf, \nf_{k+1},\mydots, \nf_{i-1}\mid\lctx\mid \tm_{i+1},\mydots,\tm_n}^{\statetwo}.\]
Thus, the statement follows from the \ih
\qed
  \end{itemize}
    \end{itemize}
\end{proof}

\gettoappendix{l:leftmost-context-decoding}
\begin{proof}
Note that:
\begin{center}
$\begin{array}{rcllllll}
\decode\state_\bullet \tsub{\name_1}{\ctxhole}
     & = &
	\left(\boapr\tsub{\name_{2}}{\decenv{\job_{2}}\genv}\mydots\tsub{\name_{n}}{\decenv{\job_{n}}\genv}
	 \right)\tsub{\name_1}{\ctxhole}
    \\
    & = &
      \boapr\tsub{\name_1}{\ctxhole}\tsub{\name_{2}}{\decenv{\job_{2}}\genv}\mydots\tsub{\name_{n}}{\decenv{\job_{n}}\genv}
      &=& \ctx^\state_{\mid\ctxhole\mid \decenv{\job_{2}}\genv,\mydots, \decenv{\job_{n}}\genv}.
    \end{array}$
\end{center}
Where $\ctx^\state_{\mid\ctxhole\mid \decenv{\job_{2}}\genv,\mydots, \decenv{\job_{n}}\genv}$ is the notation defined in the statement of \refprop{lexam-complex-invariant}. \correction{Note that $\decenv{\job_{2}}\genv, \ldots \decenv{\job_{n}}\genv$ are terms, as they are defined as the decoding of MAM states, which are always terms. Therefore, we can apply  \refprop{lexam-complex-invariant}, obtaining that} $\ctx^\state_{\mid\ctxhole\mid \decenv{\job_{2}}\genv,\mydots, \decenv{\job_{n}}\genv}$ is a non-applying leftmost context.\qed
\end{proof}

\gettoappendix{prop:leftmost-beta-proj}
\begin{proof}
The proof is exactly the same as the one for $\beta$-projection of the \EXAM (\refthm{exam-beta-projection}) except that the use of the external context read-back lemma (\reflemma{almost_pure_decoding}) is replaced by the leftmost context read-back lemma above (\reflemma{leftmost-context-decoding}).\qed
\end{proof}

\gettoappendix{coro:lexam-correctness}
\begin{proof}
Simply apply the theorem about sufficient conditions for implementations (\refthm{abs-impl}) using overhead transparency (\refprop{overhead-transparency}), which holds for any pool template, and leftmost $\beta$-projection above (\refprop{leftmost-beta-proj}).\qed
\end{proof}

\section{Evaluations to Runs}

\gettoappendix{prop:overhead-termination}
\begin{proof}
Note that if $\state$ is $\tomacho$-normal then all the jobs in (the support of) its pool (if any) have shape $\cplace\_{\la\var\tm}{\tmtwo\cons\stack}$, so that $\omeas\state=0$. Moreover, by definition $\omeas\state \geq 0$, and $\omeas\state$ is finite because by the local scope invariant (\reflemma{exam_invariant}) the environment is acyclic. Then, the statement follows from the fact that $\omeas\state$ decreases of exactly 1 for every $\tomacho$ transition, as we now prove. Cases:
\begin{enumerate}
\item \emph{Application search ($\toexamap$).}
  \[
    \begin{array}{rcl}
      \omeas{\examstate{
        \boapr 
      }{
        \cplace{\name}{\tm\,\tmtwo}{\stack}\pop\pool
      }{
        \genv
      }}
    & = &
      \omeas{\cplace{\name}{\tm\,\tmtwo}{\stack},\genv} + \omeas{\examstate\boapr\pool\genv}
    \\
    & = &
      1+\omeas{\cplace{\name}{\tm}{\tmtwo\cons\stack},\genv} + \omeas{\examstate\boapr\pool\genv}
    \\
    & = &
     1+ \omeas{\examstate{
        \boapr
      }{
        \cplace{\name}{\tm}{\tmtwo\cons\stack}\push\pool
      }{
        \genv
      }}
    \end{array}
  \]
\item \emph{Abstraction search ($\toexamlam$).}
  Let $\nametwo$ be a fresh name. Then:
  \[
    \begin{array}{rcll}
      \omeas{
        \examstate{
          \boapr
        }{
          \cplace{\name}{\lam{\var}{\tm}}{\stempty}\pop\pool
        }{
          \genv
        }
      }
     & = &
      \omeas{\cplace{\name}{\la\var\tm}{\estack},\genv} + \omeas{\examstate\boapr\pool\genv}
    \\
    & = &
      1+\omeas{\cplace{\name}{\tm}{\estack},\genv} + \omeas{\examstate\boapr\pool\genv}
    \\
    & = &
      1+ \omeas{\cplace{\name}{\tm}{\estack},\genv} + \omeas{\examstate{\boapr\tsub{\name}{\lam{\var}{\boctxhole\name}}} \pool\genv}
    \\
    & = &
	1+
      \omeas{
        \examstate{
          \boapr\tsub{\name}{\lam{\var}{\boctxhole\name}}
        }{
          \cplace{\name}{\tm}{\stempty}\push\pool
        }{
          \genv
        }
      }
    \end{array}
  \]

\item \emph{Substitution ($\toexamsub$).} If $\var\in\dom\genv$ then:
  \[
    \begin{array}{rcll}
      
      \omeas{
        \examstate{
          \boapr
        }{
          \cplace{\name}{\var}{\stack}\pop\pool
        }{
          \genv
        }
      }
    
      & = &
      \omeas{\cplace{\name}{\var}{\stack},\genv} + \omeas{\examstate\boapr\pool\genv}
          \\
      & = &
      1+\omeas{\cplace{\name}{\rename{\genv(\var)}}{\stack},\genv} + \omeas{\examstate\boapr\pool\genv}
          \\      
            & = &
      1+ \omeas{
        \examstate{
          \boapr
        }{
          \cplace{\name}{\rename{\genv(\var)}}{\stack}\push\pool
        }{
          \genv
        }
      }
    \end{array}
  \]
\item \emph{Variable search ($\toexamvar$).} Let $\var\notin\dom\genv$ and $\nametwo_1,\mydots,\nametwo_n$ are fresh names. Then:
  \[
    \begin{array}{rcll}
    &&
      \omeas{
        \examstate{
          \boapr
        }{
          \cplace{\name}{\var}{\tm_1\cons\mydots\cons\tm_n}\pop\pool
        }{
          \genv
        }
      }
    \\
    & = &
          \omeas{\cplace{\name}{\var}{\tm_1\cons\mydots\cons\tm_n},\genv} + \omeas{\examstate\boapr\pool\genv}
    \\
    & = &
          1+\sum_{i=1}^n\omeas{\cplace{\nametwo_1}{\tm_i}{\estack},\genv} + \omeas{\examstate\boapr\pool\genv}
	\\
	& = &
          1+ \sum_{i=1}^n\omeas{\cplace{\nametwo_1}{\tm_i}{\estack},\genv} + \omeas{\examstate{\boapr\tsub{\name}{\var\,\ctxhole_{\nametwo_1}\mydots\ctxhole_{\nametwo_n}}}\pool\genv}
    \\
    & = &
      1+ \omeas{
        \examstate{
          \boapr\tsub{\name}{\var\,\ctxhole_{\nametwo_1}\mydots\ctxhole_{\nametwo_n}}
        }{
          \cplace{\nametwo_1}{\tm_1}{\stempty}\cons\mydots\cons\cplace{\nametwo_n}{\tm_n}{\stempty}\add\pool
        }{
          \genv
        }
      }
    \end{array}
  \]
\end{enumerate}
\qed
\end{proof}

\gettoappendix{l:addresses-unfolding}
\begin{proof}
Simply note that the definition of $\decode\state$ is an extension of $\boapr$ by plugging terms in its holes, so all constructors of $\boapr$ but holes are in $\decode\state$ at the same address.\qed
\end{proof}

\gettoappendix{prop:beta-reflection}
\begin{proof}
Let $\state = \examstate\boapr\pool\genv$. In a $\tomacho$-normal reachable state all jobs in $\pool$ have shape $\cplace\name{\la\var\tm}{\tmtwo\cons\stack}$ for some $\name$, $\tm$, $\tmtwo$, and $\stack$ specific to the job. Note that $\pool$ cannot be empty, otherwise the state is final and by the halt property (\refprop{halt}) $\decode\state$ would be normal, against hypothesis. If $\decode\state \toext \tmtwo$ then there is an external context $\boctx$ such that $\decode\state=\boctxp{(\la\var\tm) \tm'}$,  $\tmtwo = \boctxp{\tm\isub\var{\tm'}}$. Let $\addr$ be the address of the abstraction $\la\var\tm$ taking part to the $\beta$-redex in $\decode\state$. Note that $\addr$ cannot be a defined address in $\boapr$, because by \reflemma{addresses-unfolding} all addresses in $\boapr$ are addresses in $\decode\state$ locating the same constructor, but definition of approximant there are no applied abstractions in $\boapr$. Then, $\addr$ extends the address of a hole $\ctxhole_\name$ in $\boapr$. Let $\addr=\addr_\alpha\cdot \addrtwo$ where $\addr_\alpha$ is the address of $\ctxhole_\name$ in $\boapr$. By the bijection invariant (\reflemma{exam_invariant}.3),  there is a job $\cplace\name{\la\vartwo\tmthree}{\tmfour\cons\stack}$ of name $\name$ in $\pool$. The suffix address $\addrtwo$ then is an address in the read-back of such a job, that is, in $\decst{(\la\vartwo\decenv\tmthree\env)}{\decenv\tmfour\env\cons\decenv\stack\env}$. It is easily seen that $\addrtwo$ is exactly the address of $\la\vartwo\decenv\tmthree\genv$ in $\decst{(\la\vartwo\decenv\tmthree\env)}{\decenv\tmfour\env\cons\decenv\stack\env}$, that is, that $\var=\vartwo$, $\tm = \decenv\tmthree\genv$ and $\tm' = \decenv\tmfour\genv$, because otherwise if $\addrtwo$ is the address of a sub-term in $\decenv\tmthree\genv$ or in the read-back stack then the surrounding context (both in $\decst{(\la\vartwo\decenv\tmthree\env)}{\decenv\tmfour\env\cons\decenv\stack\env}$, and in \decode\state) is not external, against the hypothesis. Then we have:
\[\examstate\boapr{\cplace\name{\la\var\tmthree}{\tmfour\cons\stack}\pop\pool}\genv \ \ \toexambeta\ \ \examstate\boapr{\cplace\name{\tmthree}{\stack}\push\pool}{\esub\var\tmfour\cons \genv}\]
and one shows that the read-back of the target state is the reduct term by reasoning as for the $\beta$-projection property.\qed
\end{proof}

\gettoappendix{coro:exam-completeness}
\begin{proof}
Simply apply the theorem about sufficient conditions for implementations (\refthm{abs-impl}) using overhead termination (\refprop{overhead-termination}) and $\beta$-reflection above (\refprop{beta-reflection}).\qed

\end{proof}

\end{document}